\documentclass[a4paper,twocolumn,superscriptaddress,11pt,accepted=2018-02-26]{quantumarticle}
\pdfoutput=1
\usepackage[utf8]{inputenc}
\usepackage[english]{babel}
\usepackage[T1]{fontenc}
\usepackage{amsmath}
\usepackage{hyperref}

\usepackage{tikz}
\usepackage{lipsum}

\usepackage{microtype}
\usepackage{bbm}
\usepackage{graphicx}
\usepackage{amsthm}
\usepackage{amssymb}
\usepackage{mathrsfs}
\usepackage{braket}
\usepackage{array}
\usepackage{centernot}

\usetikzlibrary{arrows}

\usepackage[numbers,sort&compress]{natbib}

\newcommand{\ketbra}[2]{\ket{#1}\!  \!  \bra{#2}}
\newcommand{\cCQQ}{{\rm CQQ}}
\newcommand{\cCCQ}{{\rm CCQ}}
\newcommand{\cC}{{\rm C}}
\newcommand{\qQ}{{\rm Q}}

\newcommand{\cM}{\mathcal{M}}
\newcommand{\cP}{\mathcal{P}}

\newcommand{\ic}{{\rm IC}}
\newcommand{\icpost}{{\rm IC_R}}

\newcommand{\nuparrow}{{\centernot\uparrow}}

\newtheorem{thm}{Theorem}
\newtheorem{prop}[thm]{Proposition}
\newtheorem{lemma}[thm]{Lemma}

\newtheorem{conj}[thm]{Conjecture}
\newtheorem*{thm*}{Theorem}
\newtheorem*{prop*}{Proposition}
\newtheorem*{lemma*}{Lemma}
\newtheorem*{cor*}{Corollary}
\newtheorem*{conj*}{Conjecture}
\theoremstyle{definition}
\newtheorem*{idea*}{Idea}
\newtheorem{remark}[thm]{Remark}
\newtheorem*{remark*}{Remark}
\newtheorem{example}{Example}

\newtheorem{prob}[thm]{Open Problem}

\DeclareMathOperator{\tr}{tr}

\begin{document}

\title{Non-Shannon inequalities in the entropy vector approach to
  causal structures}
\date{\today}
\author{Mirjam Weilenmann}
\email{mirjam.weilenmann@york.ac.uk}
\affiliation{Department of Mathematics, University of York,
  Heslington, York, YO10 5DD, UK.}

\author{Roger Colbeck}
\email{roger.colbeck@york.ac.uk}
\affiliation{Department of Mathematics, University of York,
  Heslington, York, YO10 5DD, UK.}

\begin{abstract}
  A causal structure is a relationship between observed variables that
  in general restricts the possible correlations between them. This
  relationship can be mediated by unobserved systems, modelled by
  random variables in the classical case or joint quantum systems in
  the quantum case. One way to differentiate between the correlations
  realisable by two different causal structures is to use entropy
  vectors, i.e., vectors whose components correspond to the entropies
  of each subset of the observed variables. To date, the starting
  point for deriving entropic constraints within causal structures are
  the so-called Shannon inequalities (positivity of entropy, conditional entropy and conditional mutual information). 
  In the
  present work we investigate what happens when non-Shannon entropic
  inequalities are included as well.  We show that in general these
  lead to tighter outer approximations of the set of realisable
  entropy vectors and hence enable a sharper distinction of different
  causal structures.  Since non-Shannon inequalities can only be
  applied amongst classical variables, it might be expected that their
  use enables an entropic distinction between classical and quantum
  causal structures.  However, this remains an open question.
  We also introduce techniques for deriving inner approximations to
  the allowed sets of entropy vectors for a given causal structure.
  These are useful for proving tightness of outer approximations or
  for finding interesting regions of entropy space.  We illustrate
  these techniques in several scenarios, including the triangle causal
  structure.
\end{abstract}

\maketitle

\section{Introduction} 
A common challenge in science is to make predictions based on
incomplete information.  Full details of the mechanism by which
correlations between two or more variables come about is often not
apparent and there may be several competing causal
explanations. Experimentation with interventions is one way to decide
between the candidate explanations~\cite{Pearl2009}.  However, in many
situations such intervention is difficult (or unethical), for instance
if certain involved systems are outside our control.

Considering a particular causal structure generally imposes
restrictions on the set of correlations that can be produced.  A
well-known example of such a constraint is a Bell
inequality~\cite{Bell1964}. That such relations can be violated using
measurements on quantum states motivates the consideration of more
general \emph{quantum} causal structures.  Correlations that can be
generated in such structures but not in their classical analogue are
the basis for several important cryptographic tasks~\cite{Ekert1991},
in particular for device-independent protocols for key
distribution~\cite{Mayers1998,Barrett2005b,Acin2006,Vazirani2014} or
the generation of private
randomness~\cite{Colbeck2009,Pironio2010,Colbeck2011,Miller2014}.
In a cryptographic scenario, an adversary is usually able to exert
influence at particular points in the protocol, which can be
conveniently encoded using a causal structure.
Characterising the set of possible classical, quantum and post-quantum
correlations within a specific causal structure provides a basis to
understand further tasks and possible quantum and post-quantum
advantages, which were initially studied in specific
cases~\cite{Clauser1969,Braunstein1988,GHZ,Cerf1997, Collins2002}.

For a general causal structure with unobserved variables, deciding
whether a given set of correlations can be generated is
computationally difficult and only feasible for small
examples~\cite{Garcia2005, Lee2015}.  One way to get around this, is
to use entropy to simplify the characterisation of the corresponding
sets of correlations~\cite{Chaves2012, Chaves2013, Fritz2013,
  Chaves2014, Chaves2014b, Henson2014, Steudel2015, Chaves2015,
  Chaves2016, Pienaar2016, Kela2017, Miklin2017}.  Rather than looking
at the distributions themselves, we consider \emph{entropy vectors}
whose components are the joint entropies of each subset of the
observed variables.
This often\footnote{For classical causal structures, the set is always
  convex, but for quantum causal structures it is not known whether
  this is always the case.} has the advantage that the set of
entropies realisable in a given causal structure is convex, in
contrast to the set of compatible distributions.  In addition, the
causal constraints can be represented by linear relations between
entropies instead of polynomial constraints.  It is also significant
that entropic constraints on possible correlations in a causal
structure are independent of the dimension of the involved random
variables. Hence, the method enables the derivation of constraints
that are valid for arbitrarily large alphabet sizes of all involved
observed and unobserved systems. These properties make entropy vectors
a convenient means to distinguish different causal structures in many
situations.

In this paper we report the use of non-Shannon inequalities for
distinguishing causal structures.  After a short outline of the
entropy vector approach and after introducing the necessary notation
in Section~\ref{sec:notation}, we go on to show in
Section~\ref{sec:non_shannon} that non-Shannon inequalities play a
central role for the distinction of causal structures.  This is
illustrated with the triangle causal structure
(Section~\ref{sec:triangle}), one of the simplest causal structures in
which there is a separation between classical and quantum at the level
of correlations.  For this example, we present numerous new entropic
constraints, which involve several infinite families of valid
inequalities, that together form the tightest entropic
characterisation of the classical triangle causal structure known to
date. This also leads us to disprove a claim that previously known
entropic approximations to this causal structure were
tight~\cite{Chaves2014, Chaves2015}.  Whether our new inequalities are
sufficient to separate classical and quantum versions of causal
structures is left as an open problem.

In Section~\ref{sec:int}, we analyse a number of other causal
structures, taking into account non-Shannon inequalities for their
entropic characterisation. These inequalities are relevant for
distinguishing different classical causal structures as well as for
settling the question of whether there is a classical-quantum
separation in the entropy vector approach.

We further analyse the role of non-Shannon inequalities for the
entropic characterisation of the causal structure relevant in the
context of information causality~\cite{Pawlowski2009} in
Section~\ref{sec:post_selected_nonShan}, where the combination of
non-Shannon inequalities with post-selection allows us to derive
numerous new entropy inequalities.

In Section~\ref{sec:inner}, we provide the first inner approximations
to the entropy cones of causal structures.  These are useful for
certifying that particular entropy vectors are realisable in a causal
structure as well as for showing tightness of an entropic outer
approximation in some cases (see Section~\ref{sec:int} for
examples). In cases where the outer approximation is not tight (or not
known to be tight), an inner approximation that shares some extremal
rays with the outer approximation allows the identification of parts
of the boundary of the true entropy cone as well as regions where
identifying the cone's boundary requires further analysis.

For comparison with the classical case, we also briefly consider
non-Shannon inequalities in the context of quantum and hybrid causal
structures in Section~\ref{sec:quantumtriangle}, which is illustrated
with the example of the triangle causal structure, before concluding
in Section~\ref{sec:conclusion}.

\section{Entropic cones and the entropy vector approach to causal structures}\label{sec:notation}
In this section, we briefly outline the entropy vector approach and
introduce the required notation. An elaborate introduction to the
topic can for instance be found in the review~\cite{our_review}.

\subsection{Entropic cones}\label{sec:classicalcone}
For a set of $n$ jointly distributed random variables
$\Omega= \left\{X_1,\ X_2,\ \ldots, X_n \right\}$ taking values in the
alphabet
$\mathcal{X}_\Omega= \mathcal{X}_1 \times \mathcal{X}_2 \times \cdots
\times \mathcal{X}_n$
we denote the set of all possible joint probability distributions as
$\mathcal{P}_n$. For a set of variables with joint distribution
$P_{\Omega} \in \mathcal{P}_n$ its \textit{Shannon
  entropy}~\cite{Shannon1948} is
$$ H(\Omega):= - \sum_{x \in \mathcal{X}_\Omega} P_{\Omega}(x) \log_2{\left(P_{\Omega}(x)\right)}.$$
The Shannon entropy of $\Omega$ and of all its subsets can be
expressed in an \textit{entropy vector} in $\mathbb{R}^{2^{n}-1}$,
\begin{equation*}
\begin{split}
{\bf H}(P):=(H(X_1),H(X_2),\ldots,H(X_n), H(X_1X_2),\\
H(X_1X_3),\ldots,H(X_1 X_2 \ldots X_n)).
\end{split}
\end{equation*}
The closure of the set of all possible entropy vectors,
$\Gamma^{*}_n$, is a convex cone, denoted as
$\overline{\Gamma^{*}_n}$~\cite{Zhang1997}.\footnote{The closure is
  taken because there isn't in general a good reason to put an upper
  bound on the alphabet sizes and it is known that
  $\overline{\Gamma^{*}_n}\neq\Gamma^{*}_n$ for
  $n\geq3$~\cite{Zhang1997}.}  While for $n\leq 3$, the \emph{entropy
  cone} $\overline{\Gamma^{*}_n}$ is
polyhedral~\footnote{\label{ft:3var} In fact $\overline{\Gamma^{*}_3}$
  equals the corresponding Shannon cone, $\Gamma_3$, introduced
  below.}~\cite{Zhang1998}, an infinite number of linear inequalities
are required to characterise $\overline{\Gamma^{*}_n}$ for
$n \geq 4$~\cite{Matus2007}. Hence, considering approximations to
$\overline{\Gamma^{*}_n}$ is common practice.

\subsubsection{Approximations to $\overline{\Gamma^{*}_n}$}
Before specifying approximations to $\overline{\Gamma^{*}_n}$, we define a few quantities, that are relevant in the following. The \textit{conditional entropy} of two disjoint subsets $X_{S},\ X_{T} \subseteq \Omega$ is defined as
$$ H(X_S | X_T):= H(X_S \cup X_T)-H(X_{T})$$ and for three mutually disjoint subsets $X_{S},\ X_{T},\ X_U \subseteq \Omega$ the \textit{conditional mutual information} of $X_S$ and $X_T$ conditioned on $X_U$ is
$$ I(X_S \!   : \!   X_T   |    X_U):= H(X_S   |   X_U)-H(X_S    |   X_{T} \cup X_U).
$$
Note that the entropy of the empty set is $H(\emptyset)=0$, so that
$H(X_S)=H(X_S | \emptyset)$, for example.
Two other entropic quantities we will make use of in this article are the \textit{interaction information}~\cite{McGill1954} of three mutually disjoint subsets $X_{S},\ X_{T},\ X_U \subseteq \Omega$,
$$ I(X_S \!   : \!   X_T \!   : \!   X_U):= I(X_S \!   : \!   X_T )-I(X_S \!   : \!   X_T | X_U),$$
and the \textit{Ingleton quantity} of four mutually disjoint subsets $X_S,\ X_T,\ X_U,\ X_V \subseteq \Omega$,
\begin{multline} \label{eq:ingleton}
I_{\rm ING}(X_S,X_T ; \!   X_U,X_V):= I(X_S \!  : \!   X_T|X_U)\\\!+\! I(X_S \!   : \!   X_T|X_V) + I(X_U \!   : \!   X_V) - I(X_S \!   : \!   X_T).
\end{multline}

For any entropy vector of a joint distribution of the random variables $\Omega$ the following \textit{Shannon inequalities} hold:
\begin{itemize}
\item \vspace{-0.1cm} For any $X_S \subseteq \Omega$, $H(X_S) \geq 0.$
\item \vspace{-0.1cm} For any disjoint $X_S, X_T \! \subseteq \! \Omega$, $H(X_S  |  X_T) \! \geq  \! 0.$
\item \vspace{-0.1cm} For any disjoint $X_S,\ X_T,\ X_U \subseteq \Omega$, ${I(X_S \!   : \!   X_T | X_U)} \geq 0.$
\end{itemize}
They are known to constrain a convex polyhedral cone, the \textit{Shannon cone}, $\Gamma_n$~\cite{Yeung1997}.
Because the Shannon inequalities hold for any entropy vector we have
$\overline{\Gamma^{*}_n} \subseteq \Gamma_n$.

The first entropy inequality that is not of Shannon type was found in~\cite{Zhang1997} and is presented in the following.
\begin{prop}[Zhang \& Yeung] \label{prop:zhangyeung} For any four discrete random variables $X_1$, $X_2$, $X_3$ and $X_4$ the following inequality holds:
\begin{align*}
I(X_1\!   : \!  X_2|X_3) +I(X_1\!   : \!  X_2|X_4)  + I(X_3\!   : \!  X_4) \\
 - I(X_1\!   : \!  X_2) + I(X_1 \!   : \!   X_3|X_2) +I(X_2 \!   : \!   X_3|X_1)\\
  + I(X_1 \!   : \!   X_2|X_3) \geq 0.
\end{align*}
In the following the lhs of this inequality is abbreviated as
$\Diamond_{X_1 X_2 X_3 X_4}$.
\end{prop}
The first account of infinite families of inequalities was given in~\cite{Matus2007}.
\begin{prop}[Mat\'{u}\v{s}] \label{prop:matusfam}
Let $X_1$, $X_2$, $X_3$ and $X_4$ be random variables and let $s \in \mathbb{N}$. Then the following inequalities hold:
\begin{align}
&s    \left[    I(X_1\!   : \!  X_2|X_3)    +    I(X_1\!   : \!
  X_2|X_4)   +   I(X_3\!   : \!  X_4)\right.\nonumber\\
&\left.\!-\!I(X_1\!   : \!  X_2)    \right]\!+\!I(X_1\!   : \!  X_3|X_2)\!+\!
        \frac{s(s\!+\!1)}{2}    \left[   I(X_2\!   : \!
        X_3|X_1)\!\right.\nonumber\\ &\left.   + I(X_1\!   : \!  X_2|X_3)   \right]   \geq   0,\label{eq:matus1}
\end{align}
\begin{align}
s&\left[    I(X_1\!   : \!  X_2|X_3)    +   I(X_1\!   : \!  X_2|X_4)    +  I(X_3\!   : \!  X_4) \right. \nonumber\\ 
&\left.  -   I(X_1\!   : \!  X_2)   \right]  +s \left[ I(X_2\!   : \!  X_3|X_1)+I(X_1\!   : \!  X_2|X_3) \right] \nonumber\\  
&+ I(X_1\!   : \!  X_3|X_2) + \frac{s(s-1)}{2} \left[ I(X_2\!   : \!  X_4|X_1) \right. \nonumber\\ 
&\left. + I(X_1\!   : \!  X_2|X_4) \right] \geq 0.\label{eq:matus2}
\end{align}
\end{prop}
For $s=1$ both inequalities are equivalent to
$\Diamond_{X_1 X_2 X_3 X_4} \geq 0$. For the current state of the art
on non-Shannon inequalities we refer to~\cite{Dougherty2011}.  To our
knowledge, all known non-Shannon entropy inequalities in four
variables that are not (known to be) rendered redundant by tighter ones
can be written as the sum of the Ingleton quantity and (conditional)
mutual information
terms~\cite{Zhang1998,Makarychev2002,Dougherty2006,Matus2007,Xu2008,Dougherty2011}.

\bigskip

Complementary to outer approximations (such as the Shannon cone,
$\Gamma_n$) it is also interesting to consider inner approximations,
$\Gamma^{I}_n$, to the $n$-variable entropy cone
$\overline{\Gamma_{n}^{*}}$. Such approximations can be defined in
terms of so-called linear rank inequalities, which are inequalities
that hold for the dimensions of subspaces of vector spaces~\cite{Ingleton}.
Entropic inequalities imply linear rank inequalities\footnote{For example, the
  linear rank inequality implied by the fact that $H(AB)\geq H(A)$ for
  all distributions $P_{AB}$ is $\dim(A\cup B)\geq\dim(A)$ for all
  subspaces $A$ and $B$ of a (finite dimensional) vector space.} but
the converse does not hold~\cite{Hammer2000}, which is why using (the entropic analogue
of) linear rank inequalities gives an inner approximation.  In the case
of $n=4$ the Shannon inequalities and the \emph{Ingleton inequality},
i.e.,
\begin{equation} I_{\rm ING}(X_1,
  X_2 ; \!  X_3, X_4) \geq 0 \label{eq:ingletoninequ}
\end{equation}
(and its permutations), define such an inner approximation
$\Gamma^{I}_4$~\cite{Ingleton}. $\Gamma^{I}_5$ is defined by the
Shannon inequalities, all instances of the Ingleton inequality and
$24$ additional classes of inequalities~\cite{Dougherty2009}. For $6$
or more variables, a complete list of all linear rank inequalities is
not known, nor is it known whether such a list would be finite. A list
of over a billion inequalities (counting permutations) has been
found~\cite{Dougherty2014}.\footnote{Note that we can always obtain
  some inner approximations with other methods, e.g., by constructing
  a set of achievable entropy vectors and taking their convex hull.}

\subsection{The entropy vector approach to causal structures}
A \textit{causal structure}, $C$, is a set of variables arranged in a
directed acyclic graph (DAG). The \emph{parents},
$X^{\downarrow_{1}}$, of a variable $X$ in a DAG are the variables
from which an arrow is directly pointing at $X$, and the
\textit{descendants}, $X^{\uparrow}$ of $X$ are all variables that may
be reached from $X$ along a directed path within the DAG. We use $C^{\cC}$ and $C^{\qQ}$ to denote the classical and the quantum version of a causal structure respectively.

\subsubsection{Classical causal structures}\label{sec:classical_method}
The graph of a classical causal structure, $C^{\mathrm{C}}$, with random variables $X_1,\ X_2,\ \ldots,\ X_n$, encodes the independence relations of $X_1,\ X_2,\ \ldots,\ X_n$ in the sense that the distribution $P_\mathrm{X_1 X_2 \ldots X_n}$ is said to be \textit{compatible} with  $C^{\mathrm{C}}$ if it can be decomposed as
$$P_\mathrm{X_1 X_2 \ldots X_n}= \prod_{i} P_\mathrm{X_{i}|X^{\downarrow_{1}}_{i}} .$$
This interpretation of classical causal structures follows the theory of Bayesian networks~\cite{Pearl2009}. The set of all compatible distributions is in the following denoted $\mathcal{P}(C^{\mathrm{C}})$.
The compatibility requirement is equivalent to the condition that for each variable $X_i$,
\begin{equation}\label{eq:indepentr}
I(X_i\!:\!X_i^{\nuparrow}|X^{\downarrow_{1}}_i)=0.
\end{equation}
$X_i^{\nuparrow}$ denotes the \textit{non-descendants} of $X_i$, i.e., all variables in the causal structure except for the variable itself and its descendants.~\footnote{In particular, all other (conditional) independence relations of variables in the causal structure are implied by these $n$ equalities.}

The entropic description of causal structures has first been considered in~\cite{Chaves2012, Fritz2013}. The $n$ equalities \eqref{eq:indepentr} restrict the $n$-variable entropy cone $\overline{\Gamma_n^{*}}$ to the cone of all entropy vectors compatible with $C^{\cC}$, denoted $\overline{\Gamma^{*}}\left(C^{\mathrm{C}}\right)$. An outer approximation to $\overline{\Gamma^{*}}(C^{\mathrm{C}})$ is constructed by supplementing $\Gamma_n$ with the same $n$ equalities, which leads to the cone $\Gamma\left(C^{\cC} \right)$.

When $k$ out of the $n$ variables of $C^{\cC}$ are
observed, we take these to be the first $k$ variables, 
$X_1, X_2, \ldots, X_k$, without loss of generality. For $k <n$ we are then interested in deriving constraints for the observed variables only.
For a compatible distribution,
$P_{{\rm X_1 X_2} \cdots {\rm X_n}} \in \mathcal{P}\left(C^{\cC}\right)$,
this is achieved by marginalising over the unobserved variables
$X_{k+1}, X_{k+2}, \ldots, X_n$ which yields a distribution in the set
of all compatible marginal distributions
$P_{{\rm X_1 X_2} \cdots {\rm X_k}} \in
\mathcal{P}_{\mathcal{M}}\left(C^{\cC}\right)$.
Entropically, marginalisation corresponds to a projection of the
entropy cone to the corresponding $k$-variable marginal cone
$\overline{\Gamma^{*}_{\mathcal{M}}}\left(C^{\mathrm{C}}\right)
\subsetneq \mathbb{R}^{2^{k}-1}$,
which would be obtained by dropping all components involving any of
the $n-k$ unobserved variables from each vector in
$\overline{\Gamma^{*}}\left(C^{\mathrm{C}}\right)$.  The outer
approximation $\Gamma\left(C^{\mathrm{C}}\right)$ can be analogously
projected to an approximation,
$\Gamma_{\mathcal{M}}\left(C^{\mathrm{C}}\right)$, of
$\overline{\Gamma^{*}_{\mathcal{M}}}\left(C^{\mathrm{C}}\right)$. Computationally,
$\Gamma_{\mathcal{M}}\left(C^{\mathrm{C}}\right)$ is usually found by considering $\Gamma\left(C^{\mathrm{C}}\right)$ characterised by means of bounding hyperplanes and applying a Fourier-Motzkin elimination algorithm to the system of linear inequalities~\cite{Williams1986}.

\subsubsection{Quantum causal structures} \label{sec:quantum_method} A
quantum causal structure $C^{\qQ}$ differs from its classical analogue
in that the unobserved nodes correspond to quantum systems. Here, we
only consider causal structures with two generations of nodes, where
the nodes of the first generation are unobserved quantum systems
and the nodes of the second generation represent observed
(classical) variables. Note that this also allows for the description
of causal structures with observed input nodes, as is illustrated in
Figure~\ref{fig:inputs}.
\begin{figure}
\centering
\resizebox{0.8 \columnwidth}{!}{%
\begin{tikzpicture}[scale=0.9]
\node (a)  at (-2,2.5) {$(a)$};
\node[draw=black,circle,scale=0.75] (X1) at (-1,2) {$X_2$};
\node[draw=black,circle,scale=0.75] (A1) at (-1,0) {$X_1$};
\node (Z1) at (-0.0,0.5) {};

\node (b) at (2,2.5) {$(b)$};
\node[draw=black,circle,scale=0.75] (X) at (3,2) {$X_2$};
\node[draw=black,circle,scale=0.75] (A) at (3,0) {$X_1$};
\node (ZZ) at (3,1) {$A$};
\node (Z) at (4,0.5) {};
\node (Y) at (5,0.5) {};

\draw [->,>=stealth] (A1)--(X1);
\draw [->,>=stealth] (ZZ)--(X);
\draw [->,>=stealth] (ZZ)--(A);
\draw [dotted,->,>=stealth] (Z1)--(X1);
\draw [dotted,->,>=stealth] (Z)--(X);

\end{tikzpicture}
}%
\caption{For a quantum causal structure with an observed input node,
  $X_1$~--~meaning a parentless node from which there is only one
  arrow to another observed node, $X_2$~--~there always exists another
  (quantum) causal structure that allows for exactly the same
  correlations and where the observed input is replaced by a shared
  quantum parent of $X_1$ and $X_2$. To simulate any correlations in
  (a) within scenario (b) we can use a quantum system that sends
  perfectly correlated classical states to both nodes $X_1$ and $X_2$,
  distributed as $X_1$. On the other hand, any correlations obtained
  in scenario (b) can be created in scenario (a) by having a random
  variable $X_1$ sent to node $X_2$, where the relevant quantum states
  (the reduced states that would be present in (b) conditioned on the
  value of $X_1$) are locally generated.  Note that these
  considerations are not restricted to quantum causal structures but
  apply also in the classical case (or even if considering states from
  a generalised probabilistic theory).  }
\label{fig:inputs}
\end{figure}
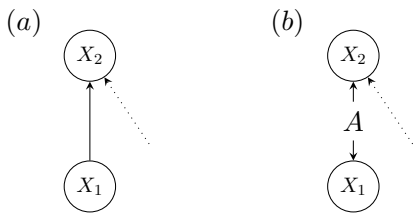

For such causal structures, each edge has an associated Hilbert space,
which can be labelled by the parent and child, e.g., for a DAG with an
edge $X \rightarrow Y$, there is an associated
$\mathcal{H}_{X_{Y}}$. Each unobserved node is labelled by a quantum
state, a density operator on the tensor product of the Hilbert spaces
associated with the edges originating at that node. For each observed
node there is an associated POVM that acts on the tensor product of
the Hilbert spaces associated with the edges that meet at that node.
The distributions,
$P\in\mathcal{P}_{\mathcal{M}}\left(C^{\qQ}\right)$, of the observed
variables that are \emph{compatible} with a causal structure
$C^{\qQ}$, are those resulting from performing the specified POVMs on
the relevant systems via the Born rule.

A technique to analyse these sets entropically was proposed by Chaves et al.~\cite{Chaves2015} and is outlined in the following, where the idea of considering entropy cones of multi-party quantum states goes back to Pippenger~\cite{Pippenger2003}. 
The set of compatible observed distributions $P \in \mathcal{P}_{\mathcal{M}}\left(C^{\qQ}\right)$ can be mapped to a set of compatible entropy vectors, the closure of which is denoted $\overline{\Gamma^{*}_{\mathcal{M}}}\left(C^{\qQ}\right)$.
To approximate this set, a system is assigned to each observed variable as well as to each outgoing edge of each unobserved node. As opposed to the classical case, where we can always define a joint distribution over all variables in a causal structure $C^{\cC}$, there is in general no joint quantum state over all systems in $C^{\qQ}$. In particular, the systems corresponding to the edges that meet at an observed node do not coexist with the outcome at that node and hence there is no joint quantum state from which a joint entropy could be derived. 
The approach is therefore based on a notion of coexistence: two systems are said to \textit{coexist} if neither
is a quantum ancestor of the other in $C^{\qQ}$, and a set of systems that pairwise coexist form a \textit{coexisting set}. For each coexisting set, $X_S \subseteq \Omega$, the von Neumann entropy $H(X_S):=- \tr(\rho_{X_S} \log_2 \rho_{X_S})$ of their joint state $\rho_{X_S}$ is defined; all of these von Neumann entropies are considered as components of an entropy vector.

For each coexisting set the entropies of all its subsets as well as
all conditional mutual informations of its systems are
positive~\cite{Lieb1973}. The conditional entropy may not be positive
in general, but for three mutually disjoint subsets of a coexisting
set, $X_T,\ X_U, \ X_V \subseteq X_S$,
$H(X_T)+H(X_U) \leq H(X_T \cup X_V)+H(X_U \cup X_V)$ holds
instead. These three types of inequality hold for the components of
any entropy vector. For the von Neumann entropy of a multi-party
quantum state no additional entropy inequalities are known. It has
been suggested, however, that any classical `balanced entropy
inequality'~\cite{Chan2003} (which includes all known non-Shannon
inequalities) may also hold for multi-party quantum
states~\cite{Cadney2012}. It is worth remarking that the lack of a
joint state for all nodes within a quantum causal structure would
restrict the applicability of such inequalities in the causal context
if they were to hold.\footnote{In the triangle causal structure, for
  instance, even if the known non-Shannon inequalities held for
  arbitrary quantum states, they would not allow us to derive any new
  entropy inequalities for the quantum version of this causal structure.}

In many circumstances the conditional entropy of certain sets of
systems is known to be positive, e.g.\ if all systems in a coexisting
set are classical. Such constraints on the entropy vectors are also
added (see~\cite{our_review} for further details).  The causal
restrictions encoded in the graph are accounted for by the condition
that two subsets of a coexisting set are independent (and hence have
zero mutual information between them) if they have no shared
ancestors.\footnote{Since a node and its (quantum) ancestors never
  coexist, conditional inependences don't have to be taken into
  account in two-generation causal structures.} To relate the
entropies of systems in different coexisting sets, data processing
inequalities (DPIs) are used: Let
$\rho_\mathrm{X_S X_T} \in \mathcal{S}(\mathcal{H}_{X_\mathrm{S}}
\otimes \mathcal{H}_{X_{\mathrm{T}}})$
and $\mathcal{E}$ be a completely positive trace preserving (CPTP) map
on $\mathcal{S}(\mathcal{H}_{X_{\mathrm{T}}})$ leading to a state
$\rho'_\mathrm{X_{\mathrm{S}} X_{\mathrm{T}}}$.  Then
\begin{equation}\label{eq:DPI}
I(X_{\mathrm{S}} \!   : \!   X_{\mathrm{T}})_{\rho'_\mathrm{X_S X_T}} \leq I(X_{\mathrm{S}} \!   : \!   X_{\mathrm{T}})_{\rho_\mathrm{X_S X_T}}.
\end{equation}
Results on the redundancy of certain DPIs have been presented in~\cite{our_review}.
All constraints on the possible entropy vectors taken together define a polyhedral cone, which we denote $\Gamma\left( C^{\qQ} \right)$. Its projection to the observed variables, $\Gamma_\mathcal{M}\left( C^{\qQ} \right)$, is an outer approximation to $\overline{\Gamma^{*}_{\mathcal{M}}}\left(C^{\qQ}\right)$, that can be computed from  $\Gamma\left( C^{\qQ} \right)$ with a Fourier-Motzkin elimination algorithm~\cite{Williams1986}.

\section{Improving current entropic characterisations with non-Shannon inequalities}\label{sec:non_shannon}
In this section we show how non-Shannon inequalities allow us to
improve the previous outer approximations to the entropy cones of
classical causal structures.  We give an improved entropic description
of the triangle causal structure of Figure~\ref{fig:3variables}(e)
(Section~\ref{sec:triangle}), discuss the application of non-Shannon
inequalities to further causal structures (Section~\ref{sec:int}) and
demonstrate that non-Shannon inequalities are also applicable in
combination with post-selection using information causality as an
example (Section~\ref{sec:post_selected_nonShan}).

The computational procedure that we use in order to derive these new
inequalities is roughly outlined in the following.  (1) We take the
Shannon inequalities for the joint distribution of all variables in a
causal structure $C^{\cC}$, (2) we add a set of valid non-Shannon
inequalities to these, (3) we add all conditional independence
equalities that are implied by $C^{\cC}$, (4) we eliminate all
entropies of unobserved variables from the full set of inequalities
(by means of a Fourier-Motzkin elimination
algorithm~\cite{Williams1986}), which leads to constraints on the
entropies of the observed variables only. 

Note that the same procedure, but missing out step (2) corresponds to
the computation of $\Gamma_{\cM}(C^\cC)$ as in~\cite{Chaves2012,
  Fritz2013, Chaves2014b} and outlined in Section~\ref{sec:classical_method}.
Thus, the inclusion of (2) is responsible for the new constraints.  In
addition to deriving entropy inequalities computationally, we also
provide analytic derivations of (infinite families of) new
inequalities.

\subsection{Improved outer approximation to the entropy cone of the
  classical triangle scenario} \label{sec:triangle} 
The triangle causal structure, called $C_3$, is one of the simplest
examples with interesting features~\cite{Branciard2012, Fritz2012,
  Chaves2014, Henson2014, Chaves2015}. It can be used when three
parties make observations, $X$, $Y$ and $Z$ respectively, on systems,
$A$, $B$ and $C$, that are shared pairwise between them. This may for
instance be realised in a communication protocol where three parties
aim to obtain (correlated) data while interacting in pairs and without
ever having interacted as a group.

$C_3$ is one of only five distinct causal structures involving three
observed random variables that exhibit no ancestral relations between
the observed variables (cf.\ Figure~\ref{fig:3variables}).
\begin{figure}
\centering
\resizebox{0.95 \columnwidth}{!}{%
\begin{tikzpicture}[scale=0.48]
\node (a)  at (-9.6,2.5) {$(a)$};
\node[draw=black,circle,scale=0.75] (X1) at (-8.5,2) {$X$};
\node[draw=black,circle,scale=0.75] (Y1) at (-4.5,2) {$Y$};
\node[draw=black,circle,scale=0.75] (Z1) at (-6.5,-1.46) {$Z$};

\node (b) at (-3.1,2.5) {$(b)$};
\node[draw=black,circle,scale=0.75] (X) at (-2,2) {$X$};
\node[draw=black,circle,scale=0.75] (Y) at (2,2) {$Y$};
\node[draw=black,circle,scale=0.75] (Z) at (0,-1.46) {$Z$};
\node (A) at (1,0.28) {$A$};

\node (c)  at (3.4,2.5) {$(c)$};
\node[draw=black,circle,scale=0.75] (X2) at (4.5,2) {$X$};
\node[draw=black,circle,scale=0.75] (Y2) at (8.5,2) {$Y$};
\node[draw=black,circle,scale=0.75] (Z2) at (6.5,-1.46) {$Z$};
\node (A2) at (6.5,0.55) {$A$};

\node (d)  at (-6.6,-3.0) {$(d)$};
\node[draw=black,circle,scale=0.75] (X3) at (-5.5,-3.46) {$X$};
\node[draw=black,circle,scale=0.75] (Y3) at (-1.5,-3.46) {$Y$};
\node[draw=black,circle,scale=0.75] (Z3) at (-3.5,-6.92) {$Z$};
\node (A3) at (-2.5,-5.18) {$A$};
\node (B3) at (-4.5,-5.18) {$B$};

\node (e)  at (0.4,-3.0) {$(e)$};
\node[draw=black,circle,scale=0.75] (X4) at (1.5,-3.46) {$X$};
\node[draw=black,circle,scale=0.75] (Y4) at (5.5,-3.46) {$Y$};
\node[draw=black,circle,scale=0.75] (Z4) at (3.5,-6.92) {$Z$};
\node (A4) at (4.5,-5.18) {$A$};
\node (B4) at (2.5,-5.18) {$B$};
\node (C4) at (3.5,-3.46) {$C$};

\draw [->,>=stealth] (A)--(Y);
\draw [->,>=stealth] (A)--(Z);

\draw [->,>=stealth] (A2)--(Y2);
\draw [->,>=stealth] (A2)--(Z2);
\draw [->,>=stealth] (A2)--(X2);

\draw [->,>=stealth] (A3)--(Y3);
\draw [->,>=stealth] (A3)--(Z3);
\draw [->,>=stealth] (B3)--(X3);
\draw [->,>=stealth] (B3)--(Z3);

\draw [->,>=stealth] (A4)--(Y4);
\draw [->,>=stealth] (A4)--(Z4);
\draw [->,>=stealth] (B4)--(X4);
\draw [->,>=stealth] (B4)--(Z4);
\draw [->,>=stealth] (C4)--(X4);
\draw [->,>=stealth] (C4)--(Y4);
\end{tikzpicture}
}%
\caption{Assuming no ancestral relations between any of the three
  observed variables $X$, $Y$ and $Z$ (i.e., no member of $\{X,Y,Z\}$
  is an ancestor of any other), the above are the only possible causal
  structures (up to relabelling). $A$, $B$ and $C$ correspond to
  unobserved variables.}
\label{fig:3variables}
\end{figure}
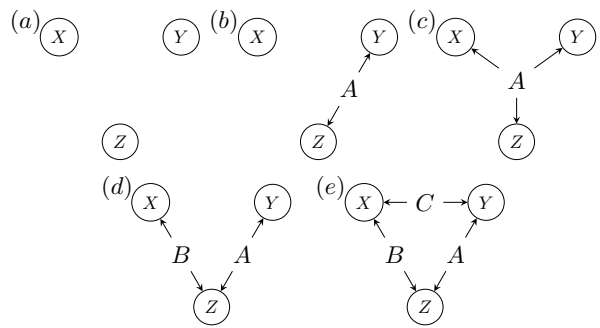
All except for the causal structures (c) and (e) may be distinguished by looking at independences among the observed variables, $X$, $Y$ and $Z$, listed in Table~\ref{table:3variablescenarios}.
\begin{table}\scriptsize
\begingroup
\def\inlinedisplayeqn#1{\vspace*{0.1ex}$\displaystyle #1$\vspace*{0.1ex}}
\begin{tabular}{|m{0.55cm}|m{4.80cm}|m{1.52cm}|}
    \hline 
    \hspace{0pt}\textbf{}  &\hspace{0pt}\textbf{Compatible Distributions} &\hspace{0pt}\textbf{Observed $ \ \ $ Independence} \\ 
		\hline
    \hspace{0pt} (a) & \hspace{-3pt}$P_{XYZ} = P_{X}P_{Y}P_{Z}$ & $I(X\! \!  :  \! \! YZ) \! \! = \! \! 0$ \newline $I(Y\! \!  : \! \! XZ)\!\! = \!\! 0$ \newline  $I(Z \! \!    :  \! \!  XY) \! \! =\! \! 0$ \\ 
 		\hline
    \hspace{0pt} (b) & \hspace{-3pt}\vphantom{{\large H}}\inlinedisplayeqn{P_{XYZ}= \sum_{A}P_{X}P_{Y| A}P_{Z| A}P_{A}} &
$I(X\! \!  : \! \!  YZ) \! \! = \! \! 0$ \\ 
 		\hline
    \hspace{0pt} (c) & \hspace{-3pt}\vphantom{{\large H}}\inlinedisplayeqn{P_{XYZ}= \sum_{A}P_{X| A}P_{Y| A}P_{Z| A}P_{A}} &
 None \\ 
 		\hline
    \hspace{0pt} (d) & \hspace{-3pt}\vphantom{{\large H}}\inlinedisplayeqn{P_{XYZ}= \sum_{A,B}P_{X| B}P_{Y\mid A}P_{Z| AB}P_{A}P_{B}} &
 $I(X\! \!  : \! \!  Y)=0$ \\ 
 		\hline
    \hspace{0pt} (e) & \hspace{-4pt}\vphantom{{\large H}}\inlinedisplayeqn{P_{X \! Y \! Z}\! \! = \! \!\!   \!   \!   \!  \sum_{A,B,C}\!   \!   \!   \!   \!   P_{X \! | \! AC}P_{Y \! | \! AC}P_{Z \! | \! AB}P_{A}P_{B}P_{C}} &
None \\ 
	\hline		
	\hline
    \end{tabular}
\endgroup
\caption{Distributions compatible with the three-variable causal structures displayed in Figure~\ref{fig:3variables}.}
\label{table:3variablescenarios}
\end{table}
However, while the causal structure of Figure~\ref{fig:3variables}(c)
does not impose any restrictions on the compatible $P_{XYZ}$, the
distributions that are compatible with the triangle causal structure
of Figure~\ref{fig:3variables}(e) obey additional
constraints~\cite{Steudel2015}.\footnote{For instance, perfectly
  correlated bits $X$, $Y$ and $Z$, i.e., those with joint
  distribution
\begin{equation} \label{eq:perfcor}
P_{XYZ}(x,y,z)= 
\begin{cases}
\frac{1}{2} &x=y=z\\
0 &\textrm{otherwise},
\end{cases}
\end{equation}
are not achievable in this causal structure.  This is not only true
classically, but also in any generalised probabilistic
theory~\cite{Steudel2015, Henson2014}.}
This illustrates that causal structures encode more than the observed independences.

Furthermore, $C_3$ is unique among these five causal structures, it
being the only one that features quantum correlations that are not
classically reproducible, i.e.,
$P_{\cM}\left(C_3^{\cC}\right) \subsetneq
P_{\cM}\left(C_3^{\qQ}\right)$,
as proven in Ref.~\cite{Fritz2012} (see
Section~\ref{sec:quantcorr} for further details regarding the quantum
scenario).\footnote{In structures (a), (b) and (c) all joint
  distributions are allowed for the variables that share a common
  cause in the classical case. Hence, quantum systems do not enable
  any stronger correlations. This is because, for any quantum
  state $\rho_A$ shared at $A$ and measured later, the correlations can
  be classically reproduced if $A$ sends out the same classical output
  statistics to the parties directly.  
  In structure (d) no non-classical quantum correlations exist
  either~\cite{Fritz2012}. This is also fairly intuitive: the quantum
  measurements performed at $X$ and $Y$ could be equivalently
  performed at the sources $B$ and $A$ respectively, such that these
  sources distribute cq-states of the form
  $\sum_{x} P_X(x) \ketbra{x}{x} \otimes \rho_{B_Z}^{x}$ and
  $\sum_{y} P_Y(y) \ketbra{y}{y} \otimes \rho_{A_Z}^{y}$ instead.  The
  same correlations can be achieved classically by taking random
  variables $B=X$ and $A=Y$ (these being distributed according to
  $P_X$ and $P_Y$).  Since $\rho_{B_Z}^{x}$ and $\rho_{A_Z}^{y}$ are
  functions of $X$ and $Y$, the statistics formed by measuring such
  states can be computed classically via a probabilistic function
  (this function could be made deterministic by taking $B=(X,W)$,
  where $W$ is distributed appropriately).}
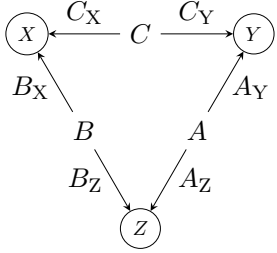
\begin{figure}
\centering
\resizebox{0.45\columnwidth}{!}{%
\begin{tikzpicture}[scale=0.8]
\node[draw=black,circle,scale=0.75]  (X) at (-2,2) {$X$};
\node[draw=black,circle,scale=0.75]  (Y) at (2,2) {$Y$};
\node[draw=black,circle,scale=0.75]  (Z) at (0,-1.46) {$Z$};
\node (A) at (1,0.28) {$A$};
\node (B) at (-1,0.28) {$B$};
\node (C) at (0,2) {$C$};

\draw [->,>=stealth] (A)--(Y) node [right,pos=0.4] {$A_\mathrm{Y}$};
\draw [->,>=stealth] (A)--(Z) node [right,pos=0.5] {$A_\mathrm{Z}$};
\draw [->,>=stealth] (B)--(X) node [left,pos=0.4] {$B_\mathrm{X}$};
\draw [->,>=stealth] (B)--(Z) node [left,pos=0.5] {$B_\mathrm{Z}$};
\draw [->,>=stealth] (C)--(X) node [above,pos=0.5] {$C_\mathrm{X}$};
\draw [->,>=stealth] (C)--(Y) node [above,pos=0.5] {$C_\mathrm{Y}$};
\end{tikzpicture}
}%
\caption{Triangle causal structure $C_3$. Three observed random
  variables $X$, $Y$ and $Z$ have pairwise common causes. In the
  classical case these common causes are random variables, $A$, $B$
  and $C$, while in the quantum case these are replaced by quantum
  systems, ($A_{\mathrm{Y}}$, $A_{\mathrm{Z}}$), ($B_{\mathrm{X}}$,
  $B_{\mathrm{Z}}$) and ($C_{\mathrm{X}}$, $C_{\mathrm{Y}}$).}
	\label{fig:triangle}
\end{figure}

In the following, we derive new and improved outer approximations
to $\overline{\Gamma^*_\cM}(C_3^{\cC})$ by using non-Shannon entropy
inequalities. These show that the Shannon approximation to
$\overline{\Gamma^*_\cM}(C_3^{\cC})$ is not tight, i.e., that
$ \overline{\Gamma^*_\cM}(C_3^{\cC}) \subsetneq
\Gamma_\cM(C_3^{\cC})$.
We remark that our findings contradict the considerations
of~\cite{Chaves2014,Chaves2015}, which together argue that in the
marginal scenario there is no separation between the Shannon
cone and the classical entropy cone, i.e., they argue that
$\Gamma_\cM(C_3^{\cC})=\overline{\Gamma^*_\cM}(C_3^{\cC})$, which
would imply that non-Shannon inequalities are
irrelevant.\footnote{The details of this are in the Supplementary
  Information of~\cite{Chaves2015}.} For further discussion of the
discrepancy with~\cite{Chaves2014,Chaves2015}, see
Appendix~\ref{sec:prevwork}.

The set of all observed distributions compatible with $C_3^{\cC}$ is\footnote{Note that this set is not convex, which can be seen by considering the perfect
  correlations $P_{XYZ}$ of~\eqref{eq:perfcor} (which are not in
  $\mathcal{P}_{\mathcal{M}}\left(C_3^{\cC}\right)$) as a convex
  combination of the distribution where $X$, $Y$ and $Z$ are always
  $0$ and the distribution where $X$, $Y$ and $Z$ are always $1$
  (both of which are).}
\begin{align*}
\mathcal{P}_{\mathcal{M}}
&\left(C_3^{\cC}\right)=\left\{P_{XYZ}\in\mathcal{P}_{3}|\phantom{\sum_{A,B,C}}\right.\\
&\left. P_{XYZ}=\sum_{A,B,C}P_{A}P_{B}P_{C}P_{X|BC}P_{Y|AC}P_{Z|AB}\right\} \! .
\end{align*}
The compatible entropy vectors are, 
\begin{equation*}
\Gamma^{*}_\mathcal{M} \! \left( \! C_3^{\cC} \! \right)\! =\! \left\{ \! v  \! \in \! \mathbb{R}^{7}_{\geq 0} | \exists P \! \in \! \mathcal{P}_\mathcal{M} \! \left( \! C_3^{\cC} \! \right) \! \text{ s.t. } \! v \! = \! {\bf{H}}(P) \! \right\} \! ,
\end{equation*}
and $\overline{\Gamma^{*}_\mathcal{M}}\left(C_3^{\cC}\right)$ is a
convex cone (cf.~\cite{our_review}).

The Shannon outer approximation\footnote{The outer approximation
  obtained from all six variable Shannon inequalities and the
  conditional independence equalities~\eqref{eq:indepentr}, which are
  in this case $I(A\!   : \!  BCX)=0$, $I(X\!   : \!  AYZ| BC)=0$ and appropriate
  permutations.},
$$\Gamma_\mathcal{M}\left(C_3^{\cC}\right)=\left\{w\in
  \Gamma_3 | M_\mathcal{M}\left(C_3^{\cC}\right) \cdot w \geq 0
\right\},$$
was explicitly computed by Chaves et al.~\cite{Chaves2014,Chaves2015},
where $M_{\mathcal{M}}\left(C_3^{\cC}\right)$ is the coefficient
matrix of the following three equivalence classes of inequalities
(where permutations of $X$, $Y$ and $Z$ lead to a total of $7$
inequalities):\footnote{Recall that an explicit linear description of
  their entropy cone is generally only available for causal structures
  with up to three nodes. In particular, such a description is not
  available for $\overline{\Gamma^{*}}\left(C_3^{\cC}\right)$, which
  involves six nodes. Hence, it is impossible to directly compute
  $\overline{\Gamma^{*}_{\mathcal{M}}}\left(C_3^{\cC}\right)$ with a
  variable elimination algorithm.}

\begin{widetext}
\begin{equation}\label{eq:margindep}
\begin{split}
- H(X)-H(Y)-H(Z)+H(XY)+H(XZ) &\geq 0, \\
-5 H(X)-5H(Y)-5H(Z)+4H(XY)+4H(XZ)+4H(YZ)-2H(XYZ) &\geq 0, \\
-3 H(X)-3H(Y)-3H(Z)+2H(XY)+2H(XZ)+3H(YZ)-H(XYZ) &\geq 0.
\end{split}
\end{equation}
\end{widetext}

We now show that tighter outer approximations of the set of achievable
entropy vectors in the marginal scenario of the triangle,
$\overline{\Gamma^{*}_{\mathcal{M}}}\left(C_3^{\cC}\right)$, can be
derived by using non-Shannon type inequalities. However, there are
infinitely many such linear entropy inequalities. To restrict the
number of inequalities to be considered, the following reasoning can
be applied. As mentioned in Section~\ref{sec:classicalcone}, all known
non-Shannon entropy inequalities for four variables can be written as
the sum of the \emph{Ingleton quantity}~\eqref{eq:ingleton} and
(conditional) mutual information terms.  Since the latter are always
positive, any non-Shannon inequality is irrelevant (i.e., implied by
existing ones) if the causal restrictions imply that the
Ingleton term is non-negative.  This significantly reduces the choices
of variable sets for which the known additional inequalities may
be relevant.

\begin{example}
Consider Proposition~\ref{prop:zhangyeung} with $(X_1,~X_2,~X_3,~X_4)=(A,~B,~C,~X)$. The corresponding inequality is
\begin{equation*}
\begin{split}
I(A\!   : \!  B|C)+I(A\!   : \!  B|X)+I(C\!   : \!  X)-I(A\!   : \!  B)\\
+I(A\!   : \!  C|B)+I(B\!   : \!  C|A)+I(A\!   : \!  B|C) \geq 0.
\end{split}
\end{equation*}
Whenever a causal structure $C^{\cC}$ implies $I(A\!   : \!  B)=0$, i.e.,
independence of $A$ and $B$, the above inequality is implied by the
Shannon inequalities and the independence constraint $I(A\!   : \!  B)=0$.
Hence it cannot improve our outer approximation.
\end{example}

The following proposition restricts the permutations of each non-Shannon inequality that may be relevant for the derivation of our improved approximations to $\overline{\Gamma^{*}_\mathcal{M}}\left(C_3^{\cC}\right)$.

\begin{prop} \label{prop:ingletonperm} 
Consider an entropy inequality on four variables
  that enforces the non-negativity of a positive linear combination of
  the Ingleton quantity~\eqref{eq:ingleton} and (conditional)
  mutual information terms. This inequality is implied by the Shannon
  inequalities 
  and the conditional independences of
 $C_3^{\cC}$  (i.e., $I(A\!   : \!   XBC)=0$, $I(X\!   : \!   YZA|BC)=0$ and appropriate permutations) for all choices of four out of the six involved random variables, except
\begin{equation*}
\begin{split}
\left( X_1,~X_2)~(X_3,~X_4 \right) &= \left( X,~Y)~(Z,~C \right)\\
&= \left( X,~Z)~(Y,~B \right) \\
&= \left( Y,~Z)~(X,~A \right),
\end{split}
\end{equation*}
up to exchange of $X_1$ and $X_2$ or exchange of $X_3$ and $X_4$. 
\end{prop}
All known irredundant non-Shannon inequalities satisfy the conditions of this proposition. Note also that the application of non-Shannon inequalities to subsets of four out of the six random variables in $C_3^{\cC}$ does not encompass all possible applications of these inequalities. Specifically, each inequality can also be applied to sets of five or to all six random variables, where the joint distribution of some sets of two or three random variables are interpreted as those of one of the four random variables in the non-Shannon inequality. We have not looked into such configurations. 

\begin{proof}
For four random variables $X_1$, $X_2$, $X_3$ and $X_4$, the Ingleton inequality
\begin{eqnarray}
I(X_1\! \!  : \! \! X_2|X_3) \! + \! I(X_1\! \!  : \! \! X_2|X_4)&& + I(X_3\! \!   : \! \!  X_4) \nonumber\\
&&- \! I(X_1\! \!  : \! \! X_2)\geq0,\label{eq:ingletproof}
\end{eqnarray}
can be equivalently rewritten in four more ways with the following equalities:
\begin{alignat}{2} \label{eq:rewritings}
\! I(X_1\! \!   : \! \!  X_2|X_3) \! - \! I(X_1\! \!  : \! \!  X_2)& \! = \! I(X_1\! \!  : \! \! X_3|X_2)&& \! - \! I(X_1\! \!  : \! \!  X_3) \nonumber  \\
& \! = \! I(X_2\! \!  : \! \!  X_3|X_1)&& \! - \! I(X_2\! \!  : \! \! X_3), \nonumber \\
\! I(X_1\! \!  : \! \! X_2|X_4) \! - \! I(X_1\! \!  : \! \! X_2)& \! = \! I(X_1X_4|X_2)&& \! -\! I(X_1\! \!   : \! \!  X_4) \nonumber \\
&\! = \! I(X_2\! \!  : \! \! X_4|X_1)&& \! - \! I(X_2\! \!  : \! \!  X_4) . 
\end{alignat}
For the inequality~\eqref{eq:ingletproof} not to be implied by the Shannon inequalities and the conditional independences we need $X_1$, $X_2$, $X_3$ and $X_4$ to be such that
\begin{alignat}{3}\label{eq:infoterms}
&I(X_1\!   : \!  X_2)& &> 0, \nonumber \\ 
&I(X_1\!   : \!  X_3)& &> 0, \nonumber \\
&I(X_1\!   : \!  X_4)& &> 0, \\
&I(X_2\!   : \!  X_3)& &> 0, \nonumber \\
&I(X_2\!   : \!  X_4)& &> 0, \nonumber 
\end{alignat}
hold simultaneously. If the conditional independences of $C_3^{\cC}$
imply that one of these mutual informations is zero then the Ingleton
inequality can be expressed as a positive linear combination of
(conditional) mutual information terms in one of its five equivalent
forms and the corresponding non-Shannon inequality is redundant.

For the five constraints \eqref{eq:infoterms} to hold simultaneously,
$X_1$ and $X_2$ have to be correlated with one another as well as with
two further variables. This excludes the independent sources $A$, $B$
and $C$ as candidates for $X_1$ and $X_2$; therefore
$X_1,~X_2 \in \left\{X,~Y,~Z \right\}$.  Furthermore, the variables
$X_3$ and $X_4$ have to be correlated with both, $X_1$ and $X_2$. This
excludes the two variables in $\left\{ A,~B,~C \right\}$ that do not
lie between $X_1$ and $X_2$ in $C_3^{\cC}$. Hence, for each choice of
$X_1$ and $X_2$, the variables $X_3$ and $X_4$ have to be chosen as
the remaining element of $\left\{X,~Y,~Z \right\}$ and the variable
positioned opposite it in $C_3^{\cC}$.

In summary, $\left(X_1,~X_2 \right) \left(X_3,~X_4 \right)$ can only be
$\left(X,~Y \right) \left(Z,~C \right)$, $\left(X,~Z \right)
\left(Y,~B \right)$ and $\left(Y,~Z \right) \left(X,~A \right)$ up to
permutations of the variables within a tuple.
\end{proof}

If we were to take one 4-variable non-Shannon inequality into account
and apply it to any subset of four out of the total of six random
variables in the causal structure, this would leave us with $360$
permutations of the inequality (if the inequality is not invariant
under the permutation of any of the four involved
variables). Proposition~\ref{prop:ingletonperm} reduces this to only
$12$ (potentially) irredundant permutations.

For each non-Shannon inequality, these $12$ permutations are
candidates for improving the outer approximation to
$\overline{\Gamma^{*}_{\mathcal{M}}}\left(C_3^{\cC}\right)$. We remark
here that for most known non-Shannon inequalities, several of these
$12$ permutations can be shown to be redundant\footnote{For instance,
  if the non-Shannon inequality in question is invariant under the
  permutation of some of its variables then some of the $12$
  permutations are equivalent, or, if the marginalisation of different
  permutations of the same inequality (that are not equal) imply the
  same inequalities for the marginal scenario then some of these
  inequalities may be redundant for our purposes.}.  Despite
accounting for this reduction in the permutations of each inequality,
the number of different inequalities to be considered is infinite,
and any outer approximation to
$\overline{\Gamma^{*}_\mathcal{M}}\left(C_3^{\cC}\right)$ could
(potentially) be tightened further by including additional
inequalities.

In principle, the more inequalities that are added, the better the
approximation to
$\overline{\Gamma^{*}_\mathcal{M}}\left(C_3^{\cC}\right)$. However,
adding too many inequalities at a time renders the task of
marginalising infeasible. Applied to a system of $n_0$ inequalities
the Fourier-Motzkin algorithm can yield up to
$\left(\frac{n_0}{2}\right)^{2}$ inequalities in the first elimination
step. Iterating the procedure for $n$ steps produces up to
$4 \cdot \left(\frac{n_0}{4}\right)^{2^{n}}$ inequalities. To avoid
this double exponential behaviour the elimination algorithm can be
adapted by implementing a few rules to remove some of the many
redundant inequalities produced in each step. These rules are
collectively known as \u{C}ernikov
rules~\cite{Cernikov1960,Chernikov1965} and comprehensively explained
in~\cite{Bastrakov2015}. It is known, however, that the number of
necessary inequalities can still grow
exponentially~\cite{Monniaux2010}.  That said, the worst case scaling
may not be exhibited in our case. In fact, the inequalities defining
$\Gamma\left(C_3^{\cC}\right)$ contain few variables each and thus
lead to far fewer than the maximal number of inequalities. However,
computational resources still limit us to adding a relatively small
number of different supplementary inequalities to the standard Shannon
cone at a time.

We have used the previously outlined technique to compute tighter
outer approximations to
$\overline{\Gamma_\cM^*}\left(C_3^{\cC}\right)$, by including a
manageable number of non-Shannon inequalities at a time:\smallskip

\noindent \emph{Case 1}: We include the inequality
from Proposition~\ref{prop:zhangyeung} as well as all six inequalities
from~\cite{Dougherty2006} applied to all subsets of four out of the
six variables of $C_3^{\cC}$. This leads to $45$ classes of
inequalities, of which $41$ are not part of the outer approximation $\Gamma_{\mathcal{M}}\left(C_3^{\cC}\right)$.\smallskip

\noindent \emph{Case 2}: We include the inequalities of the form given
in~\eqref{eq:matus1} and~\eqref{eq:matus2} for $s=1,2,3$ and for all
subsets of four out of the six variables in $C_3^{\cC}$. In this case,
we find $114$ classes of inequalities, of which $110$ are not part of
the outer approximation
$\Gamma_{\mathcal{M}}\left(C_3^{\cC}\right)$.\smallskip 

\noindent In each case, all classes (together with the number of
members in each class) are provided as Supplementary
Information. \smallskip

\noindent 
We have compared our new approximations to the Shannon outer
approximation by sampling uniformly over the surface of the positive
sector of the unit hypersphere around ${\bf 0}$ in
$\mathbb{R}^7$~\cite{Muller1959}\footnote{I.e., from the set
  $\{{\bf v}\in\mathbb{R}^7:v_i\geq 0, \sum_{i=1}^7v_i^2=1\}$.}. A
measure for the hyperdimensional solid angle included by these
approximations is given in terms of the fraction, $\alpha$, of points within the
respective cones.  We have sampled $3.2\times10^{9}$ points
each, which led to the following estimates for $\alpha$: \smallskip

\noindent \emph{Shannon Cone}: $\alpha_S=(3.308 \pm 0.010)\times10^{-5}$.\smallskip

\noindent \emph{Case 1}: $\alpha_1=(3.090 \pm 0.010)\times10^{-5}$.\smallskip

\noindent \emph{Case 2}: $\alpha_2=(3.072 \pm 0.010)\times10^{-5}$. \smallskip

This shows that the difference between the three approximations it relatively small: the hyperdimensional solid angle encompassed by the cones of the Case~1 and Case~2 approximations are both roughly $93 \%$ of that of the Shannon cone. An explicit entropy vector that lies in the Shannon approximation, but not in either of the new outer approximations to $\overline{\Gamma^{*}_\mathcal{M}}\left(C_3^{\cC}\right)$ is
$\left( \!   H(X) \!  , \! H(Y) \!  , \! H(Z) \!  , \! H(X\! Y) \!  , \! H(X\! Z) \!  , \! H(Y \! Z) \!  , \! H(X \! Y \! Z) \!   \right)\\
=\left( 11, 14, 14, 20, 20, 23, 28 \right)$.
We also derive some valid families of inequalities.
\begin{prop}\label{lemma:matus1}
All entropy vectors $v\in\Gamma^*_\mathcal{M}\left(C_3^{\cC}\right)$ obey
\begin{align}
&\!\left(-\frac{1}{2} s^2 - \frac{3}{2} s \right) \left( H(X) +H(Z) \right)- \left(s+1\right) H(Y)\nonumber\\
+&\!\left( \frac{1}{2} s^{2} + \frac{3}{2} s + 1 \right) \left( H(XY) + H(YZ) \right)\nonumber\\
+&s(s+2)H(XZ)-\left(s+1\right)^2 H(XYZ) \! \geq \! 0,\label{eq:marginalmatus1}\\ \nonumber\\
&\!\left( \!  \! - \frac{1}{2} s^2 \!  - \!  \frac{3}{2} s \!  - \! 2 \!  \right) \!  \left( H(X) \!  + \!  H(Y) \!  + \!  H(Z) \!  - \!  H(XY) \right) \nonumber\\
+&\! \left(\frac{1}{2} s^2 + \frac{3}{2} s +1\right)H(XZ)
+\left(  s + 2 \right) H(YZ) \nonumber\\ 
-&\! \left( s+1 \right) H(XYZ) \geq 0,\label{eq:mat1fam2}
\end{align}
\begin{align}
&\! \left(-\frac{1}{2} s^2 - \frac{3}{2} s -2 \right) \left( H(X) +H(Z)-H(XY) \right)\nonumber\\
-& \! \left(2s+2\right) H(Y)+\left( s^{2} + 2 \right) H(XZ) \nonumber\\
+& \! \left( \!  \frac{1}{2} s^{2} \!  + \!  \frac{3}{2} s \!  + \!  1 \!  \right)   H(YZ)
- \left( \!  s^2 \! + \! 1 \!  \right) H(XYZ) \! \geq \!  0,\label{eq:mat2}
\end{align}
for all $s\in\mathbb{N}$. The same holds for all permutations of $X$, $Y$ and $Z$.
\end{prop}

The proof of this proposition can be found in Appendix~\ref{sec:families}. 

Further families of inequalities can be derived by separately considering different inequalities from a family, e.g.\ the same permutation of~\eqref{eq:matus1} for each $s\in\mathbb{N}$, and combining them with the same Shannon inequalities to obtain new constraints on the marginal scenario by means of the Fourier-Motzkin elimination algorithm. Tighter inequalities are often obtained by combining several permutations of an inequality~\eqref{eq:matus1}.

Combining instances of~\eqref{eq:matus1} for several $s\in\mathbb{N}$
leads to an even larger number of new inequalities, which render many
of the families derived with the previously explained method
redundant. For the few orders $s$ up to which we were able to run our
calculations, the families~\eqref{eq:marginalmatus1}
and~\eqref{eq:mat1fam2} from Proposition~\ref{lemma:matus1} were the
only two for which none of the inequalities were implied by others.
Similar considerations can be applied to \eqref{eq:matus2} (from
which~\eqref{eq:mat2} is derived) and to further families of
inequalities~\cite{Dougherty2011}.

One might imagine that adding genuine five and six variable
inequalities to $\Gamma\left(C_3^{\cC}\right)$ leads to further
entropy inequalities for $C_3^\cC$. It turns out that applying the
five and six variable inequalities from~\cite{Zhang1998,
  Makarychev2002} to five and six variables of the triangle causal
structure respectively does not lead to a tighter outer approximation
to $C_3^{\cC}$ than the inequality from Proposition~\ref{prop:zhangyeung}. This can be shown by expanding the inequalities into a
linear combination of mutual information terms and applying a similar
reasoning to that in the proof of
Proposition~\ref{prop:ingletonperm}. As they are not particularly
instructive, the technical details of these arguments are omitted
here. The same is not known to hold for the inequality derived
in~\cite{Zhang2003}.

\begin{conj}\label{conj:matusconj}
  Infinitely many linear inequalities are needed to characterise
  $\overline{\Gamma^{*}_{\cM}}\left(C_3^{\cC}\right)$.
\end{conj}

Our main evidence for this is that the families of
inequalities~\eqref{eq:matus1}, used by Mat\'{u}\v{s} to prove that
the analogue of this conjecture holds for $\overline{\Gamma^{*}_4}$,
lead to infinite families of inequalities for $C_3^{\cC}$ after
marginalising (cf.\ Proposition~\ref{lemma:matus1}).  The curve
constructed by Mat\'{u}\v{s} in Ref.~\cite{Matus2007} to prove his
statement for $\overline{\Gamma^{*}_4}$ can be adapted to our
scenario, which can be used to show that the
inequalities~\eqref{eq:marginalmatus1} are independent.  However, we
were not able to show that this curve can be realised with entropy
vectors that are compatible with the triangle causal structure, and
hence we cannot exclude the possibility that the marginal cone
$\overline{\Gamma^{*}_{\cM}}\left(C_3^{\cC}\right)$ is polyhedral.

The infinite families of inequalities (cf.\
Proposition~\ref{lemma:matus1}) that we obtained from Mat\'{u}\v{s}'s
original family of inequalities may indicate that this region of
entropy space retains a non-polyhedral segment after the causal
constraints are included and the set is projected to the marginal
scenario. However, it could be that non-polyhedral boundary regions do
not survive the mapping to entropy vectors for $C_3^{\cC}$.  If this
were the case then (most of) our infinite set of inequalities would be
rendered redundant by another inequality.

\subsection{Application of non-Shannon inequalities to various causal structures} \label{sec:int}

The concept of a generalised DAG was introduced in~\cite{Henson2014},
the idea being to have a framework in which classical, quantum and
even more general systems can be shared by unobserved nodes.  For the
details, we refer to the original paper. The part that is of interest
here is that the authors of Ref.~\cite{Henson2014} list $21$
generalised DAGs with up to six nodes for which there may be a
separation between the correlations realisable classically and quantum
mechanically, i.e., between $\cP_{\cM}(C^\cC)$ and
$\cP_{\cM}(C^\qQ)$~\cite{Henson2014,Pienaar2016}.\footnote{Note that
  there are further causal structures with $5$ and $6$ nodes that have
  this property, which can all be reduced to these $21$ examples with
  rules specified in~\cite{Henson2014}.}  We analyse
these from an entropic perspective, looking for a causal structure $C$
in which there is a separation between
$\overline{\Gamma^*_\cM}\left(C^{\cC}\right)$ and
$\overline{\Gamma^*_\cM}\left(C^{\qQ}\right)$.
Among these structures there are three that have fewer than six nodes, displayed in Figure~\ref{fig:HensonStructures}. 
\begin{figure}
\centering
\resizebox{0.98 \columnwidth}{!}{%
\begin{tikzpicture}[scale=0.58]
\node (a)  at (-10,1) {$(a)$};
\node[draw=black,circle,scale=0.75]  (Z1) at (-5,-2) {$Y$};
\node[draw=black,circle,scale=0.75]  (Y1) at (-7,-2) {$Z$};
\node[draw=black,circle,scale=0.75]  (X1) at (-9,-2) {$X$};
\node (A1) at (-6,0.5) {$A$};
\node (b) at (-3,1) {$(b)$};
\node[draw=black,circle,scale=0.75]  (W) at (-2,-2) {$W$};
\node[draw=black,circle,scale=0.75]  (X) at (-1,0.5) {$X$};
\node (A) at (0,-2) {$A$};
\node[draw=black,circle,scale=0.75]  (Y) at (1,0.5) {$Y$};
\node[draw=black,circle,scale=0.75]  (Z) at (2,-2) {$Z$};
\node (c)  at (3.5,1) {$(c)$};
\node[draw=black,circle,scale=0.75]  (Y2) at (4.5,0.5) {$Y$};
\node[draw=black,circle,scale=0.75]  (Z2) at (6.5,0.5) {$Z$};
\node[draw=black,circle,scale=0.75]  (X2) at (5.5,-0.75) {$X$};
\node (A2) at (4.5,-2) {$A$};
\node (B2) at (6.5,-2) {$B$};
\node (00) at (7.5,-2) { };
\draw [->,>=stealth] (A1)--(Y1);
\draw [->,>=stealth] (A1)--(Z1);
\draw [->,>=stealth] (Y1)--(Z1);
\draw [->,>=stealth] (X1)--(Y1);
\draw [->,>=stealth] (W)--(X);
\draw [->,>=stealth] (A)--(X);
\draw [->,>=stealth] (A)--(Y);
\draw [->,>=stealth] (Z)--(Y);
\draw [->,>=stealth] (A2)--(X2);
\draw [->,>=stealth] (B2)--(Z2);
\draw [->,>=stealth] (B2)--(X2);
\draw [->,>=stealth] (A2)--(Y2);
\draw [->,>=stealth] (X2)--(Z2);
\draw [->,>=stealth] (X2)--(Y2);
\end{tikzpicture}
}%
\caption{Three causal structures, $C$, for which the outer approximation, $\Gamma_\cM(C^{\cC})$ tightly approximates the classical entropy cone $\overline{\Gamma^{*}_\cM}(C^{\cC})$, which also coincides with $\overline{\Gamma^{*}_\cM}(C^{\qQ})$. The observed variables are labelled $W$, $X$, $Y$ and $Z$, the unobserved nodes are called $A$ and $B$.}
\label{fig:HensonStructures}
\end{figure}
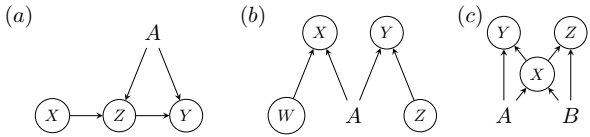
For these three, we find that the vertices of the corresponding
Shannon cone, $\Gamma_\cM(C^{\cC})$, are achievable with entropy
vectors of classical probability distributions compatible with the
causal structure, from which it follows that this cone is equal to the
entropy cone $\overline{\Gamma^*_\cM}(C^{\cC})$. (This can also be
shown by computing an inner approximation to the corresponding entropy
cones and showing that the inner and outer approximations coincide,
e.g.\ by employing linear rank inequalities as outlined in
Section~\ref{sec:inner}.)  Our results also imply that the
consideration of non-Shannon inequalities cannot lead to any further
constraints in these three causal structures.  In the following, we
furthermore show that there is no entropic separation between
classical and quantum versions of these causal structures.

\begin{prop}\label{prop:hen}
Let $C$ be any of the causal structures shown in
Figure~\ref{fig:HensonStructures}.  Then
$\overline{\Gamma^*_\cM}(C^{\cC})=\overline{\Gamma^*_\cM}(C^{\qQ})$.
\end{prop}
\begin{remark}
  Note that there are causal structures involving up to five variables
  that reduce to those shown in Figure~\ref{fig:HensonStructures}
  under the reduction rules from~\cite{Henson2014}.  Our proof does
  not rule out that these exhibit a classical to quantum
  separation.
\end{remark}
\noindent Further details, including the proof of Proposition~\ref{prop:hen} are given in Appendix~\ref{sec:uptofive}.

The $18$ remaining example causal structures involve six
variables. For all of them we have found that several instances of the
non-Shannon inequality from Proposition~\ref{prop:zhangyeung} lead to
tighter entropic constraints for the classical marginal scenarios than
those listed in~\cite{Henson2014}.

For the causal structures with four observed variables, instances of
this inequality are relevant even without considering the unobserved
nodes. These instances thus hold whether or not the unobserved nodes
are classical or quantum.  Hence, they allow us to tighten the outer
approximations to the sets of achievable entropy vectors in both
cases, in contrast to non-Shannon inequalities that are applied
to unobserved variables classically (for which the quantum analogue is not
known to hold).\footnote{Note that this reasoning is not restricted to
  distinguishing classical and quantum, but, it may also apply to the
  comparison of different causal structures with the same set of
  observed variables. While non-Shannon inequalities derived from
  unobserved variables may lead to a separation between the two causal
  structures, non-Shannon inequalities valid only for the observed
  variables may not.}

The above considerations have not enabled us to show a
separation between the achievable entropy vectors in the classical and
quantum cases, hence we are left with the following open problem.

\begin{prob}\label{prob:oooo}
  Find a causal structure $C$ with a set of observed nodes $\cM$ in
  which the sets $\overline{\Gamma^*_\cM}(C^\cC)$ and
  $\overline{\Gamma^*_\cM}(C^{\qQ})$ are provably different, or show
  that this can never occur.
\end{prob}

\subsection{Application of non-Shannon inequalities with post-selection}\label{sec:post_selected_nonShan}
In the discussion so far we have not considered a related technique
that allows for post-selection on particular outcomes of certain
variables.  The idea of doing this first appeared
in~\cite{Braunstein1988} based on results by Fine~\cite{Fine1982,
  Fine1982a} and was later generalised~\cite{Chaves2012, Fritz2013,
  Chaves2013, Chaves2015, Pienaar2016, Chaves2016}.  We refer
to~\cite{our_review} for an explanation of this technique.

Here we illustrate that non-Shannon inequalities can be used in
combination with post-selection by discussing a specific example
relevant for information causality~\cite{Pawlowski2009}. Information
causality is an information theoretic principle obeyed by classical
and quantum physics but not by general probabilistic theories in which
there are correlations that violate Tsirelson's
bound~\cite{Cirelson1980}, e.g.\ generalized no signalling
theory~\cite{Barrett07}, which allows PR-Boxes as a
resource~\cite{Tsirelson1993, Popescu1994FP}. The principle is stated
in terms of the optimal performance of two parties in a game, which we
describe below, and is quantified in terms of an entropic quantity.

Alice holds two pieces of information\footnote{In general the game is
  formulated for more, but we restrict to two here for simplicity.},
$X_0$ and $X_1$, she can send classical information $Z$ to Bob, who is
later given a message $R$ indicating whether he should guess $X_0$ or
$X_1$.  Bob's guess is denoted $Y$.  Alice and Bob are able to use a
pre-shared resource (depicted as $A$) to help them. The relevant
causal structure of the game is displayed in
Figure~\ref{fig:IC_inner}(a) and it is often analysed after
post-selecting on the value of $R$, which can be done using the
causal structure of Figure~\ref{fig:IC_inner}(b) (note that in the
quantum case the variables $Y_{|R=0}$ and $Y_{|R=1}$ do not coexist,
so it doesn't make sense to consider $\ic^\qQ_{\mathrm{R}}$; instead a
restricted set of entropies needs to be considered -- see later).  A
theory is said to obey information causality if for all pre-shared
resources allowed by the theory,
$I(X_0\! \!   : \! \! Y_{|R=0}) \! + \! I(X_1\! \!  : \! \!  Y_{|R=1}) \leq H(Z)+I(X_0\!   : \!  X_1)$.

A stronger set of entropic constraints for this causal structure were
found in~\cite{Chaves2015}, including the relation
\begin{align}
I(X_0\!   : \!  Z Y_{\rm |R=0}) \! +  \! &I(X_1\!   : \!  Z Y_{\rm |R=1})  \! +  \! I(X_0\!   : \!  X_1 | Z Y_{\rm |R=1}) \nonumber\\ 
&\leq H(Z) +I(X_0 \!   : \!   X_1) \, , \label{eq:IC_known}
\end{align}
which holds for both classical and quantum shared
resources.\footnote{Because the existence of a joint distribution of
  $Y_{\rm | R=0}$ and $Y_{\rm | R=1}$ with appropriate marginals is
  not clear in the quantum case, the two variables have
  to be interpreted as alternatives and are part of different
  coexisting sets. Therefore, the analysis of $\ic_{\rm R}^\cC$ does not carry over to the quantum case, but a separate analysis is required there.}

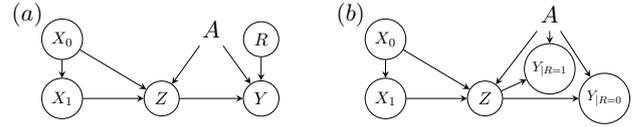
\begin{figure}
\centering 
\resizebox{1\columnwidth}{!}{%
\begin{tikzpicture} [scale=0.8]
\node (A0) at (-8.2,1.5) {$(a)$};
\node[draw=black,circle,scale=0.7] (A1a) at (-7.5,1) {$X_0$};
\node[draw=black,circle,scale=0.7] (A1) at (-7.5,-0.2) {$X_1$};
\node[draw=black,circle,scale=0.75] (A2) at (-5.5,-0.2) {$Z$};
\node[draw=black,circle,scale=0.75] (A3) at (-3.5,-0.2) {$Y$};
\node[draw=black,circle,scale=0.75] (A5) at (-3.5,1) {$R$};
\node (A4) at (-4.5,1.2) {$A$};
\draw [->,>=stealth] (A1)--(A2);
\draw [->,>=stealth] (A1a)--(A2);
\draw [->,>=stealth] (A2)--(A3);
\draw [->,>=stealth] (A4)--(A2);
\draw [->,>=stealth] (A4)--(A3);
\draw [->,>=stealth] (A5)--(A3);
\draw [->,>=stealth] (A1a)--(A1);

\node (B0) at (-1.7,1.5) {$(b)$};
\node[draw=black,circle,scale=0.7] (B1a) at (-1,1) {$X_0$};
\node[draw=black,circle,scale=0.7] (B1) at (-1,-0.2) {$X_1$};
\node[draw=black,circle,scale=0.75] (B2) at (1,-0.2) {$Z$};
\node (B4) at (2.3,1.5) {$A$};
\node[draw=black,circle,scale=0.6] (B3) at (3.4,-0.2) {$Y_{|R=0}$};
\node[draw=black,circle,scale=0.6] (B3a) at (2.3,0.4) {$Y_{|R=1}$};

\draw [->,>=stealth] (B1)--(B2);
\draw [->,>=stealth] (B1a)--(B2);
\draw [->,>=stealth] (B1a)--(B1);
\draw [->,>=stealth] (B4)--(B2);
\draw [->,>=stealth] (B4)--(B3);
\draw [->,>=stealth] (B4)--(B3a);
\draw [->,>=stealth] (B2)--(B3);
\draw [->,>=stealth] (B2)--(B3a);
\end{tikzpicture}
}%
\caption{(a) Causal structure underlying the Information Causality
  game, $\ic$. Alice holds a database, here made up of two bits $X_0$
  and $X_1$. These need not be independent, which is expressed by a
  potential causal influence from $X_0$ to $X_1$. She is then allowed
  to send a message $Z$ to Bob, who, depending on which bit $R$ a
  referee asks for, takes a guess $Y$ of either $X_0$ or $X_1$. Alice
  and Bob may have shared some resources (represented by $A$) before
  performing the protocol, either some classical randomness, a quantum
  system, or a resource from a more general non-signalling theory, which Alice may
  use in order to choose her message and Bob may use to make his
  guess.  (b) The effective causal structure of the Information
  Causality game after post-selecting on binary $R$, labelled $\icpost$. This
  causal structure shares some of its marginal distributions with
  conditional distributions of $\ic$, i.e., if we use $P$ for the
  distribution in $\icpost$ and $Q$ for that in $\ic$ then
  $P_{X_0 X_1 Z Y_{| R=r}}=Q_{X_0 X_1 Z Y | R=r}$ for $r=0,1$.}
\label{fig:IC_inner}
\end{figure}

We show that using non-Shannon inequalities leads to a tighter outer
approximation of the information causality scenario in the case of a
classical shared resource.  Considering just the inequality from
Proposition~\ref{prop:zhangyeung} (and permutations) has led us to
derive a total of $265$ classes of entropy inequalities, including the
$52$ classes that were obtained without non-Shannon constraints
in~\cite{Chaves2015} (a list of all $265$ classes together with the
number of representatives of each class is available as Supplementary
Information). Moreover, we expect further non-Shannon inequalities to
lead to numerous additional constraints potentially rendering our
inequalities redundant. In principle, infinite families of
inequalities, similar to those found in Proposition~\ref{lemma:matus1}
for the triangle scenario could also be derived here.

In the quantum case, we can only apply the non-Shannon inequalities to
the two coexisting sets of exclusively classical variables
$\left\{ X_0, X_1, Z , Y_{\rm | R=0} \right\}$ and
$\left\{ X_0, X_1, Z , Y_{\rm | R=1} \right\}$, which means that we
can impose a set of $24$ additional constraints (including
permutations) just by adding all permutations of the inequality from
Proposition~\ref{prop:zhangyeung} to the outer approximation that is
obtained without these (no further variable elimination is required).
 
It is worth pointing out that although our results (in the form of new
inequalities) imply that previous entropic characterisations of
$\icpost$ were not tight, the inequality~\eqref{eq:IC_known} is not
rendered redundant by our new inequalities.

\section{Inner approximations to the entropy cones of causal structures}\label{sec:inner}

To complement the outer approximations, it is sometimes useful to
consider inner approximations to the entropy cones of causal
structures. This is particularly useful when one can show that inner
and outer approximations coincide, as they then identify the actual
boundary of the entropy cone. Examples for this are the three causal
structures of Figure~\ref{fig:HensonStructures}, also discussed in the
previous section.  Hence, inner and outer approximation together serve
as a relatively simple means to identify the boundary of certain
entropy cones. Such findings also immediately imply that non-Shannon
inequalities are irrelevant for improving on the outer approximation
to the entropy cone for the causal structure in question.

Furthermore, we can often find inner approximations that share
extremal rays with the outer approximations derived from the Shannon
and independence constraints (even when the two do not coincide). They
hence allow us to identify the regions of entropy space where our
approximations are tight and those regions where there is a gap
between inner and outer approximation.\footnote{Such a comparison of
  inner and outer approximations can be performed for the entropy cone
  of a causal structure including its unobserved variables, i.e.,
  before marginalisation, as well as for the respective approximations
  to its marginal cone, which we are mainly interested in here.}  Such
a gap can be explored, e.g. by using non-Shannon inequalities, as was
explained in the previous section.

Inner approximations also serve as a tool to decide whether entropy
vectors are suitable for certifying the unattainability of particular
distributions that are suspected not to be achievable within the
causal structure at hand. If such a distribution leads to an entropy
vector within an inner approximation to the entropy cone in question,
this means either that the distribution is in fact achievable within
the causal structure or that the causal structure allows for another
distribution with the same entropy vector (or an arbitrarily good
approximation of such).  Hence, to determine whether the distribution
in question is achievable, switching to a more fine-grained method
(see for example~\cite{Pienaar2016, Wolfe2016}) is
necessary.

In the following we show how inner approximations can be found in
different scenarios.

\subsection{Techniques to find inner approximations for causal structures with up to five observed variables}
For a causal structure, $C$, that involves a total of four or five
variables, inner approximations to its entropy cone can be derived
from $\Gamma_4^{I}$ or $\Gamma_5^{I}$ respectively (as defined in
Section~\ref{sec:classicalcone}) combined with the conditional
independence constraints of $C^{\cC}$, which together constrain a cone
$\Gamma^{I}\left(C^{\cC}\right)$. An inner approximation to the
corresponding marginal scenarios,
$\Gamma^{I}_{\cM}\left(C^{\cC}\right)$, is then obtained from
$\Gamma^{I}\left(C^{\cC}\right)$ with a Fourier-Motzkin elimination,
like for outer approximations. It is guaranteed that
$\Gamma_\mathcal{M}^{I}\left(C^{\cC}\right)$ is an inner approximation
to $\overline{\Gamma_\mathcal{M}^{*}}\left(C^{\cC}\right)$, as it is a
projection of an inner approximation
$\Gamma^{I}\left(C^{\cC}\right) \subseteq
\overline{\Gamma^{*}}\left(C^{\cC}\right)$.
Hence, inner approximations can be straightforwardly computed for such
causal structures. Examples where this applies are the three causal
structures of Figure~\ref{fig:HensonStructures}.

\begin{example}[Inner approximation to the instrumental
  scenario.] \label{example:innerI} For the classical instrumental
  scenario, $C_{\rm I}$ of Figure~\ref{fig:HensonStructures}(a), we
  can compute an inner approximation by adding the conditional
  independence constraints $I(A\!   : \!  X)=0$ and $I(X\!   : \!  Y|AZ)=0$ to the
  Ingleton cone $\Gamma^I_4$, as prescribed above. We can, however,
  also directly prove that
  $\overline{\Gamma^{*}_{\cM}}\left(C_{\rm
      I}^{\cC}\right)=\Gamma_{\cM}\left(C_{\rm I}^{\cC}\right)$
  by showing that all permutations of the Ingleton inequality are implied by Shannon and
  conditional independence constraints and, hence, inner and outer
  approximations coincide for $C_{\rm I}^{\cC}$.  Since $I(A\!   : \!  X)=0$,
  $I_{\rm ING}\left(A,X; \! Y,Z \right) \geq 0$ is immediately implied by Shannon and
  independence constraints. Furthermore, the rewritings of $I_{\rm ING}$
  according to \eqref{eq:rewritings} imply that $I(A\!   : \!  X)=0$
  which (together with the Shannon inequalities) implies all permutations of the
  Ingleton inequality except for
  $I_{\rm ING}\left(Y,Z; \! A,X \right) \geq 0$. We can rewrite
\begin{align*}
&I_{\rm ING}\left(Y,Z; \! A,X \right)\\
&= \! I(Y\! \!  : \! \!  Z|A) \!   + \!   I(Y\! \!  : \! \!  X|Z) \!   + \!   I(X\! \!  : \! \! A|Y) \!   - \!   I(X\! \!  : \! \! Y|A) \\
&= \! I(Y\! \!  : \! \! X|Z) \!   + \!   I(X\! \!  : \! \!  A|Y) \!   + \!   I(Y\!  \! : \! \!  Z|A \! X) \! \!  - \! \!  I(X\! \!  : \! \! Y|A \! Z),
\end{align*}
the positivity of which is hence implied by the Shannon inequalities and the independence constraint ${I(X\!   : \!  Y|AZ)=0}$.
\end{example}

Implementing all relevant linear rank inequalities of four and five
variables (which includes their permutations and the application of
the Ingleton inequality to each four variable subset as well as
grouping several variables to one)~\cite{Dougherty2009} and then
performing a variable elimination may be impractical and
computationally challenging for certain causal structures.
Furthermore, for causal structures that involve more than five nodes
not all possible linear rank inequalities are known and their number
may even be infinite~\cite{Dougherty2014}. It is therefore useful to
derive inner approximations by other methods.  For a causal structure,
$C$, the following methods are examples of how to derive inner
approximations, $\Gamma_{\mathcal{M}}^{I}\left(C^{\cC}\right)$:
\begin{itemize}
\item Construct (random) entropy vectors from distributions compatible with $C^{\cC}$ and take their convex hull.
\item Take the vertices of $\Gamma_{\mathcal{M}}\left(C^{\cC}\right)$ that are reproducible with distributions compatible with the causal structure, their convex hull is an inner approximation. 
\item Take the outer approximation to the classical causal structure, $\Gamma\left(C^{\cC}\right)$, as a starting point and add a manageable number of linear rank inequalities to derive further constraints. These inequalities may be employed either before or after marginalising, which leads to different cones.\footnote{If for instance all linear rank inequalities in up to $k$ observed variables are added after marginalisation, the resulting cone corresponds to the intersection $\Gamma_{\mathcal{M}}\left(C^{\cC}\right) \cap \Gamma^{I}_{k}$, where $k$ is the number of observed variables.} The convex hull of the reproducible rays is an inner approximation.
\end{itemize}

For the three examples of Figure~\ref{fig:HensonStructures} it is rather straightforward to recover all extremal rays of the outer approximation to the marginal scenario, $\Gamma_{\mathcal{M}}\left(C^{\cC}\right)$ (cf.\ also Appendix~\ref{sec:uptofive}), i.e., the second method above is effective.

Overall, we found that whenever the extremal rays are not all straightforwardly recovered, the third method is effective. This is our preferred technique because by starting out with extremal rays of the Shannon cone we obtain approximations that in some regions are already tight (as opposed to the first method), and, at the same time adding linear rank inequalities helps us identify those extremal rays that are likely to be reproducible with distributions in $C^{\cC}$ (this may help us avoid dropping reproducible rays in some situations). 
The entropy cones obtained in this way are not necessarily inner approximations, and, if they are, they have to be proven as such, for example by explicitly constructing distributions that reproduce entropy vectors on each of the extremal rays (as with the second method above). However, in all our examples this method allowed us to recover a cone of which all extremal rays were easily seen to be reproducible after adding only few linear rank inequalities to $\Gamma\left(C^{\cC}\right)$. 
(If this were not the case one could still drop several irreproducible rays from the resulting cones to obtain an inner approximation.) 
The method is illustrated in the example below.

We also remark here that in order to improve on inner approximations obtained with the second or third method above, the first method is applicable.

\begin{example}\label{example:IC_inner}
  Consider the classical causal structure of
  Figure~\ref{fig:IC_inner}(a) and remove the node $R$ to give a
  $5$-variable causal structure, $\hat{\rm IC}^{\cC}$. We can in
  principle consider all linear rank inequalities of five random
  variables combined with all Shannon inequalities and the conditional
  independence constraints, which would give us an inner
  approximation, $\Gamma^{I}_{\cM}\left(\hat{{\rm IC}}^{\cC}\right)$,
  to the entropy cone,
  $\overline{\Gamma^{*}_{\cM}}\left(\hat{\rm IC}^{\cC}\right)$.  This
  procedure would involve a (impractically) large number of
  inequalities.

  Instead, we can consider the outer approximation in terms of Shannon
  inequalities and conditional independence constraints,
  $\Gamma_{\cM}\left(\hat{\rm IC}^{\cC}\right)$, and intersect this
  cone with the Ingleton cone for the four observed variables,
  $\Gamma^{I}_4$, i.e., we add all permutations of the Ingleton
  inequality for the four observed variables to
  $\Gamma_{\cM}\left(\hat{\rm IC}^{\cC}\right)$.  This is easily
  obtained but does not result in any restrictions beyond those of the
  Shannon outer approximation, which is characterised by $52$ extremal
  rays.

  Adding the Ingleton inequality for all subsets of four out of the
  five random variables to $\Gamma\left(\hat{\rm IC}^{\cC}\right)$
  before performing the variable elimination, only $46$ extremal rays
  are recovered. These are straightforward to reproduce with entropy
  vectors in $\hat{\rm IC}^{\cC}$.\footnote{We can show that the $6$
    extremal rays of the Shannon cone that are not part of this inner
    approximation are not achievable in $\hat{\rm IC}^{\cC}$, because they
    violate the entropy inequalities we obtain when taking non-Shannon
    inequalities into account in the computation of the outer
    approximations to
    $\overline{\Gamma^{*}_{\cM}}\left(\hat{\rm IC}^{\cC}\right)$.} A
  detailed exposition of this is presented in
  Appendix~\ref{sec:example}.
\end{example}

This method can also be applied to causal structures with more than five variables. 
For the 
first few causal structures
from~\cite{Henson2014} we have recovered inner approximations by adding the Ingleton inequality to $\Gamma_{\cM}\left(C^{\cC}\right)$, i.e., by taking the intersection $\Gamma_{\cM}\left(C^{\cC}\right) \cap \Gamma_k^{I}$ (the extremal rays as well as distributions
recovering entropy vectors on each of them are available
as Supplementary Information). 

In the following we give a detailed analysis of the inner approximation to the triangle causal structure and compare this to the outer approximations presented in previous sections.

\subsection{Example: Inner approximation to $\overline{\Gamma^{*}_\mathcal{M}}\left(C_3^{\cC}\right)$}
Here, we derive an inner approximation to the entropy cone compatible with $C_3^{\cC}$. An inner approximation to $\overline{\Gamma^{*}_6}$ in terms of linear rank inequalities is not available (see also Section~\ref{sec:classicalcone}).
Nonetheless, we are able to derive an inner approximation to
$\overline{\Gamma^{*}_{\cM}}\left(C_3^{\cC}\right)$ by relying on
Ingleton's inequality. In the following, we apply
\eqref{eq:ingletoninequ} to any subset of four out of the six random
variables of $C_3^{\cC}$ and take all their permutations into account.
We concisely write these inequalities in a matrix $M_\mathrm{I}$ and
consider the cone
\begin{equation*}
\Gamma^\mathrm{I}\! \left(C_3^{\cC}\right) \! := \! \left\{v \in \Gamma_6 | M_\mathrm{CI}\! \left(C_3^{\cC}\right) \cdot v = 0, \ M_\mathrm{I} \cdot v \geq 0 \right\} \! .
\end{equation*}
When marginalising this cone we obtain
\begin{equation*}
\Gamma_\mathcal{M}^\mathrm{I}\left(C_3^{\cC}\right):= \left\{w \in \Gamma_3 | M_{\mathrm{I},\mathcal{M}}\left(C_3^{\cC}\right) \cdot w \geq 0 \right\},
\end{equation*}
where $M_{\mathrm{I},\mathcal{M}}\left(C_3^{\cC}\right)$ contains only
one inequality,\footnote{Inequality~\eqref{eq:intinfo} renders the
  three Shannon inequalities of the form $I(X\!   : \!  Y|Z)\geq 0$
  redundant. $\Gamma_\mathcal{M}^\mathrm{I}\left(C_3^{\cC}\right)$ is
  thus fully characterised by the six remaining three variable Shannon
  inequalities (constraining $\Gamma_3$) and \eqref{eq:intinfo}. }
\begin{equation}\label{eq:intinfo}
-I(X\!   : \!  Y\!   : \!  Z) \geq 0.
\end{equation}
This relation can also be analytically derived from the Ingleton
inequality and the conditional independence constraints of
$C_3^{\cC}$. \footnote{The proof proceeds as follows. There are only
  three instances of the Ingleton inequality that are not implied by
  the conditional independences and the Shannon inequalities (cf.\
  also Proposition~\ref{prop:ingletonperm}). The independence
  constraint ${I(X\!   : \!  Y|C)=0}$ and its permutations ${I(X\!   : \!  Z|B)=0}$ and
  ${I(Y\!   : \!  Z|A)=0}$ lead to \eqref{eq:intinfo} in all three cases.}

\begin{prop} \label{prop:triangleingleton}
$\Gamma_\mathcal{M}^\mathrm{I}\left(C_3^{\cC}\right)$ is an inner approximation to the marginal entropy cone of the triangle causal structure,  
\begin{equation*} 
\Gamma_\mathcal{M}^\mathrm{I}\left(C_3^{\cC}\right)\subsetneq \overline{\Gamma^*_\mathcal{M}}\left(C_3^{\cC}\right).
\end{equation*}
\end{prop}
The proof of Proposition~\ref{prop:triangleingleton} is deferred to Appendix~\ref{sec:inner_appendix}. $\Gamma_\mathcal{M}^\mathrm{I}\left(C_3^{\cC}\right)$ provides a certificate for vectors to be realisable as entropy vectors in $C_3^{\cC}$: if a vector $v \in \mathbb{R}^{7}$ obeys all Shannon constraints as well as~\eqref{eq:intinfo}, then it lies in
$\overline{\Gamma_\mathcal{M}^*}\left(C_3^{\cC}\right)$. Compared to
the different outer approximations to
$\overline{\Gamma_\mathcal{M}^*}\left(C_3^{\cC}\right)$ analysed in
Section~\ref{sec:non_shannon} the hyperdimensional solid angle for
this inner approximation is considerably smaller. Sampling over the
unit hypersphere around ${\bf 0}$ in
$\mathbb{R}^7$ as before (meaning $3.2\times10^9$ samples), we obtain $\alpha_I=(2.147\pm0.008)\times10^{-5}$.

It is worth emphasising that not all correlations whose entropy
vectors lie in $\Gamma_\mathcal{M}^\mathrm{I}\left(C_3^{\cC}\right)$
can be realised in $C_3^{\cC}$.  Instead, if ${\bf
  H}(P)\in\Gamma_\mathcal{M}^\mathrm{I}$ then there exists
$P'\in\cP_\cM(C_3^{\cC})$ such that ${\bf H}(P')={\bf H}(P)$. The correlations of Figure~\ref{fig:fritzproof}, realised in the quantum version of the triangle causal structure, $C_3^{\qQ}$, which will be considered in detail in Section~\ref{sec:quantcorr}, are one such example.
\label{sec:w_correl}
These are not in $\cP_\cM(C_3^{\cC})$, but their entropy
vector nevertheless satisfies~\eqref{eq:intinfo}. Our argument
implies that there must be another distribution realisable in $C_3^{\cC}$ with the same entropy vector.\footnote{Another such example is given by the distribution
$$P_{XYZ}(x,y,z)=\left\{\begin{array}{cl}\frac{1}{3}&\text{ }
    (x,y,z)=(1,0,0),(0,1,0),(0,0,1)\\
    0&\text{ otherwise,}\end{array}\right.$$ which is not compatible
with $C_3^{\cC}$, as shown in~\cite{Wolfe2016}.}

\section{Non-Shannon inequalities in the quantum and hybrid triangle causal structures} \label{sec:quantumtriangle} 

In this section, we compare classical and quantum versions of the triangle
causal structure (the distinction reflecting the nature of the
unobserved nodes). We also consider hybrid scenarios, in which some of the unobserved systems are restricted to be classical
while others are quantum.  These turn out to be insightful for
understanding the gap between classical and quantum causal structures.
We also analyse whether non-Shannon inequalities lead to improved entropic
characterisations in these cases.

\subsection{Quantum triangle scenario}\label{sec:quantcorr}
It was first shown in Ref.~\cite{Fritz2012}, that there are joint distributions among the three observed variables $X$, $Y$ and $Z$ in $C_3^{\qQ}$ that cannot be reproduced in $C_3^{\cC}$, based on the CHSH scenario (see Figure~\ref{fig:fritzproof} and Appendix~\ref{sec:fritzproof} for the details).
Hence $C_3^{\qQ}$ might also lead to a larger set of compatible entropy vectors than $C_3^{\cC}$.

Entropically, $C_3^{\qQ}$ can be analysed with the technique outlined
in Section~\ref{sec:quantum_method}. An outer approximation,
$\Gamma_\mathcal{M}\left(C_3^{\qQ}\right)$, to the set of achievable
entropy vectors,
$\overline{\Gamma_\mathcal{M}^{*}}\left(C_3^{\qQ}\right)$, was
constructed in~\cite{Chaves2015}. It led to the Shannon
inequalities for the jointly distributed $X$, $Y$ and $Z$ and the
additional inequality
\begin{equation}
I(X\!   : \!  Y) + I(X\!   : \!  Z) \leq H(X),\label{eq:fritzhenson}
\end{equation}
as well as its permutations in $X$, $Y$ and $Z$~\cite{Chaves2015}.

It is natural to ask whether tighter approximations to
$\overline{\Gamma_\mathcal{M}^{*}}\left(C_3^{\qQ}\right)$ can be
realised by a similar procedure to the one that led to tighter
approximations in the classical case.  Unfortunately, we don't know of
any similar inequalities for the von Neumann entropy of multi-party
quantum states.  Furthermore, even if the known non-Shannon
inequalities were to hold for von Neumann entropy we would not be able
to use them to add constraints to $C_3^{\qQ}$ due to the lack of large
enough sets of coexisting, interdependent variables.\footnote{Note
  that in causal structures that involve four or more observed
  variables non-Shannon inequalities can be applied to these. However,
  non-Shannon inequalities cannot be applied to unobserved quantum
  systems (see also Section~\ref{sec:post_selected_nonShan}).}

\begin{prob}\label{op:1}
Do the closures of the sets of compatible entropy vectors coincide in the classical and the quantum triangle scenario, i.e., does $\overline{\Gamma^{*}_{\mathcal{M}}}\left(C_3^{\cC}\right)=\overline{\Gamma^*_{\mathcal{M}}}\left(C_3^{\qQ}\right)$ hold?
\end{prob}
Note that if this were to be answered in the affirmative, it would
point towards deficiencies of the current entropic techniques for
approximating
$\overline{\Gamma^*_{\mathcal{M}}}\left(C_3^{\qQ}\right)$, which are
not able to recover any additional inequalities similar to the
non-Shannon inequalities found in the classical case.

One way to solve this problem would be to find an entropy vector
compatible with $C_3^{\qQ}$ that lies outside one of our outer
approximations to
$\overline{\Gamma^{*}_{\mathcal{M}}}\left(C_3^{\cC}\right)$. Random
searches where the sources $A$, $B$ and $C$ distribute up to four
qubits each did not yield violations.  However, the evidence from
these random searches against a separation of the classical and the
quantum sets is relatively weak.  For one, our classical outer
approximations might be so loose that they contain
$\overline{\Gamma^*_{\mathcal{M}}}\left(C_3^{\qQ}\right)$.  To counter
this, we have attempted to randomly search for vectors that lie in
$\overline{\Gamma^*_{\mathcal{M}}}\left(C_3^{\qQ}\right)$ but not in
the classical inner approximation
$\Gamma_\mathcal{M}^\mathrm{I}\left(C_3^{\cC}\right)$.  In spite of
the fact that we know such vectors exist, we were unable to randomly
find any.  This shows the weakness of random searching, and also that
the region we are looking for (if it even exists) is small with
respect to our sampling measure.\footnote{This is not a statement
  about the geometric extent of this region (for instance in terms of
  a hyperdimensional solid angle as was previously considered for
  inner and outer approximations to
  $\overline{\Gamma^*_{\mathcal{M}}}\left(C_3^{\cC}\right)$). Instead,
  since we are sampling quantum states here, and since these are not
  in a one-to-one correspondence with the entropy vectors, it is a
  statement about the fraction of states and measurements that may
  produce entropy vectors outside
  $\Gamma_\mathcal{M}^\mathrm{I}\left(C_3^{\cC}\right)$ (in low
  dimensions) according to our sampling distribution. This must be a
  very small proportion of states and measurements (we didn't sample
  any).  Note also that if there is a gap between
  $\overline{\Gamma^*_{\mathcal{M}}}\left(C_3^{\cC}\right)$ and
  $\overline{\Gamma^*_{\mathcal{M}}}\left(C_3^{\qQ}\right)$ it is
  smaller than that between
  $\Gamma^{\rm I}_{\mathcal{M}}\left(C_3^{\cC}\right)$ and
  $\overline{\Gamma^*_{\mathcal{M}}}\left(C_3^{\qQ}\right)$. Hence,
  constructing a vector in the first gap by sampling quantum states is
  even more difficult than for the second.}

A natural candidate for an entropy vector that might violate some of our classical inequalities
is the one corresponding to the CHSH correlations that were shown not to be reproducible in $C_3^{\cC}$ in Ref.~\cite{Fritz2012} (detailed in Figure~\ref{fig:fritzproof} where
$Z=(A',B')$ and in Appendix~\ref{sec:fritzproof}). However, the corresponding entropy vector lies inside $\Gamma_{\mathcal{M}}^{I}\left(C_3^{\cC}\right)$ so is classically
reproducible.  This particular distribution is also achievable in the
causal structure $P_4$ (a causal structure equivalent to the one in
Figure~\ref{fig:HensonStructures}(b)).  Any distribution compatible
with $P_4$ may be realised in $C_3$ by choosing one of the variables,
e.g.\ $Z$, to have two outputs, one depending only on the input from
node $A$ and the other one depending on the input from $B$.
Distributions realisable in $P_4^{\qQ}$ or $P_4^{\cC}$ are thus always
realisable in $C_3^{\qQ}$ or $C_3^{\cC}$ respectively.  According to
the results of~\cite{pnpaper}, all entropy vectors realised with
distributions in $P_4^{\qQ}$ are also classically achievable, i.e.,
realisable in $P_4^{\cC}$ (at least asymptotically).
Hence, no distribution in $P_4^{Q}$ can violate any of the classical entropy inequalities valid for $C_3^{\cC}$. A way that might still allow us to use our knowledge about quantum correlations that are not classically reproducible in the Bell scenario to violate our entropic constraints to $\overline{\Gamma^*_{\mathcal{M}}}\left(C_3^{\cC}\right)$, is by processing the inputs to all three nodes $X$, $Y$ and $Z$, so as to get around the results from~\cite{pnpaper}.\footnote{Two distributions that share the same entropy vector can
  be very different and hence may be separated by local processing.}

In the following, we generalise the distribution that was utilised in Ref.~\cite{Fritz2012} to show that there is a separation between the achievable distributions in $C_3^{\cC}$ and $C_3^{\qQ}$, to a scenario where there is local processing at each output node. This also allows us reduce the required dimension of the output at $Z$ for which one can provably detect a difference between classical and quantum distributions from two bits to one bit.
\begin{prop} \label{prop:andextension}
There are non-classical quantum correlations in $C_3$ in the case where $X$ and $Y$ output two bits each while $Z$ outputs only one. 
\end{prop}

A proof of Proposition~\ref{prop:andextension} can be found in
Appendix~\ref{sec:fritzproof}. It is interesting in so far
as the example in~\cite{Fritz2012} relies on a Bell inequality
violation. Given this, one might have expected that all information about the
measurement choices in the Bell setup, $\tilde{A}$ and $\tilde{B}$,
has to be exposed at the observed node $Z$.
Proposition~\ref{prop:andextension} shows that this is not the
case. 

\begin{figure}
\centering
\resizebox{0.6 \columnwidth}{!}{%
\begin{tikzpicture}
\node (X) at (-2,2) {$(\tilde{X},\tilde{B})$};
\node (Y) at (2,2) {$(\tilde{Y},\tilde{A})$};
\node (Z) at (0,-1.46) {$Z$};
\node (A) at (1,0.28) {$A$};
\node (B) at (-1,0.28) {$B$};
\node (C) at (0,2) {$C$};
\draw [->,>=stealth] (A)--(Y);
\draw [->,>=stealth] (A)--(Z);
\draw [->,>=stealth] (B)--(X);
\draw [->,>=stealth] (B)--(Z);
\draw [->,>=stealth] (C)--(X);
\draw [->,>=stealth] (C)--(Y);
\end{tikzpicture}
}%
\caption{Scenario involving unobserved quantum systems, leading to a
  distribution which is not reproducible with classical $A$, $B$ and
  $C$~\cite{Fritz2012}.   The observed variables $X=(\tilde{X},\tilde{B})$ and $Y=(\tilde{Y},\tilde{A})$ are chosen such that
$P_{\mathrm{\tilde{X}\tilde{Y}|AB}}$ maximally violates the CHSH
inequality~\cite{Clauser1969}.
 $Z=(A',B')$ is such that $B'=\tilde{B}=B$ and
$A'=\tilde{A}=A$.
In essence the reason that this cannot be realised in the causal
structure $C_3^{\cC}$ is the CHSH violation. Note though that it is
also important that information about $A$ is present in both $Y$ and
$Z$ (and analogously for $B$), otherwise the correlations could be mocked up. 
In Proposition~\ref{prop:andextension}, we prove that a strategy where
$Z=\mathrm{AND}(A',B')$ also leads to correlations that cannot be
classically realised (see Appendix~\ref{sec:fritzproof} for further details).}
\label{fig:fritzproof}
\end{figure}
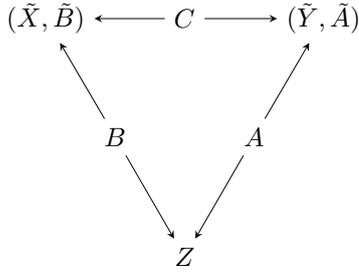

Nonetheless, we find that the entropy vector used to prove this proposition
does not violate our classical inequalities.  
We have also taken $Z$ to be determined by different functions of $A$ and $B$ and have
additionally considered local processing of $X$ and $Y$.  However,
even after such post-processing, for instance by applying all possible
functions from two bits to one, we have not been able to detect any
violations of the classical entropic bounds.

Note that vectors outside
$\Gamma_{\mathcal{M}}^{I}\left(C_3^{\cC}\right)$ can be constructed
with appropriate post-processing of the (quantum) distribution. A
possible way to achieve this is applying \textsc{and} or \textsc{or}
functions appropriately. One may for instance consider the quantum
scenario detailed above, and take
$X=\operatorname{AND}(\tilde{X}, \tilde{B})$,
$Y=\operatorname{AND}(\tilde{Y}, \tilde{A})$ and
$Z=\operatorname{OR}(A',B')$. This renders the interaction information
of the entropy vector of the joint distribution of $X$, $Y$ and $Z$
positive, so the vector is not in
$\Gamma_{\mathcal{M}}^{I}\left(C_3^{\cC}\right)$.

We have similarly tried to violate our entropy inequalities by relying on games other than the CHSH scenario, for which we know that there is distinctive quantum
behaviour (i.e., a separation at the level of
correlations); these include input states and
measurements known to lead to violations of the chained Bell
inequalities~\cite{Braunstein1988} or the Mermin-Peres magic square
game~\cite{Mermin1990, Peres1990}, all with post-processing at ($X$, $Y$ and) $Z$.  

We have further considered scenarios where all three parties measure entangled states and use the
measurement outputs as inputs for further measurements. We have also attempted to incorporate functions known to lead to a positive interaction information in the classical case, as well as
functions from two to one bits in general, into these scenarios. None
of these attempts has led to a violation of the classical
inequalities so far. In a number of scenarios we have also considered
shared PR-boxes instead of entangled states, again without detecting
any violations of the inequalities. In most cases the corresponding
entropy vectors have a negative interaction information, and hence lie
in $\Gamma_\mathcal{M}^\mathrm{I}\left(C_3^{\cC}\right)$, so can be
realised with a classical distribution as well, like in the case of
the correlations mentioned at the end of Section~\ref{sec:w_correl}.
\label{sec:op}

\subsection{Hybrid triangle scenarios}
In a \emph{hybrid causal structure} some of the unobserved nodes are
allowed to be quantum, whereas others are restricted to be
classical. One motivation for this is that sharing entanglement over
large distances is challenging due to noise, so two distant
observations might be assumed to have a classical cause while nearby
ones could have quantum causes.  In the case of the causal structure $C_3$,
there are two such hybrid scenarios: either one or two of the three
unobserved variables can be chosen to be classical, whereas the others
are quantum. We call these two causal structures $C_3^{\cCQQ}$ and
$C_3^{\cCCQ}$ respectively. In the following, we will approximate the
sets of compatible entropy vectors for both scenarios.  We show that
in hybrid scenarios of the triangle causal structure non-Shannon
inequalities are relevant.

\subsubsection{$C_3^{\cCQQ}$ scenario}
In this scenario one of the unobserved variables is classical (we take
this to be $A$).  The techniques introduced in
Sections~\ref{sec:classical_method} and~\ref{sec:quantum_method} allow
us to compute approximations of the set of allowed entropy vectors. We
find
\begin{equation*}
\Gamma_\mathcal{M}\left(C_3^{\cCQQ}\right)=\Gamma_\mathcal{M}\left(C_3^{\qQ}\right),
\end{equation*}
i.e., the outer approximation to
$\overline{\Gamma_\mathcal{M}^{*}}\left(C_3^{\cCQQ}\right)$ obtained
without taking non-Shannon inequalities into account coincides with
the outer approximation to
$\overline{\Gamma_\mathcal{M}^{*}}\left(C_3^{\qQ}\right)$.  However,
unlike in the fully quantum case $C_3^{\qQ}$, non-Shannon constraints
can be included for $C_3^{\cCQQ}$, for instance the
inequality from Proposition~\ref{prop:zhangyeung} with variable
choices
\begin{align*}
\Diamond_{YZAX} \geq 0,\quad  \Diamond_{YZXA} \geq 0. 
\end{align*}
This results in a tighter approximation to
$\overline{\Gamma_\mathcal{M}^{*}}\left(C_3^{\cCQQ}\right)$, which comprises the Shannon inequalities for three variables,
the constraint~\eqref{eq:fritzhenson} and\footnote{The second of
  these inequalities can be easily derived from $\Diamond_{YZXA} \geq
  0$ and the conditional independences, analogously to
 Proposition~\ref{lemma:matus1}. To derive the first inequality, on the
  other hand, several inequalities have to be combined.}
\begin{align*}
-3 H(X)&-3 H(Y)- 3 H(Z)+ 2H(XY)\\
&+ 2H(XZ)+ 3H(YZ)- H(XYZ) \geq 0, \\
-2 H(X)&-2 H(Y)- 2 H(Z)+ 3H(XY)\\
&+ 3H(XZ)+ 3H(YZ)- 4H(XYZ) \geq 0.
\end{align*}
Further non-Shannon constraints could also be exploited to improve
these approximations. Hence, some of the extremal rays of the Shannon
outer approximation $\Gamma_{\cM}\left( C_3^{\qQ} \right)$ are
provably not achievable if $A$, $B$ and $C$ do not all share entangled
states. Note that this does not imply that the sets of achievable entropy vectors in 
$C_3^{\qQ}$ and $C_3^{\cCQQ}$ differ. However, the difference in their outer approximations 
may prove useful for analysing whether there is a difference between the two. 

If one were to prove by other means that there is no such difference the inequalities for $C_3^{\cCQQ}$ would give us a way to better approximate the set of achievable entropy vectors of $C_3^{\qQ}$.

\subsubsection{$C_3^{\cCCQ}$ scenario}
In this scenario we take $A$ and $B$ to be classical. This scenario
can be understood as a Bell scenario, where the measurement choices of
the two parties are unobserved and processed to one single observed
output, $Z$.\footnote{Note that even though the sets of achievable
  entropy vectors in classical and quantum case coincide in the Bell
  scenario (cf.\ Section~\ref{sec:int} and~\cite{pnpaper}) this may
  not be the case here as very different distributions may lead to the
  same entropy vector in the classical and quantum case, which may be
  separated by local processing.} The distributions from
Section~\ref{sec:quantcorr}, that are provably not reproducible in
$C_3^{\cC}$ can be generated in this causal structure. Its entropic analysis thus
restricts the violations of our classical inequalities we may hope to achieve with such distributions.
To approximate the set of compatible entropy vectors of this scenario,
$\overline{\Gamma^{*}_{\mathcal{M}}}\left(C_3^{\cCCQ}\right)$, we
proceed analogously to the $C_3^{\qQ}$ and $C_3^{\cCQQ}$ scenarios
before. However, the result differs and leads to a tighter cone, even
without considering non-Shannon inequalities, i.e., 
$\Gamma_{\mathcal{M}}\left(C_3^{\cCCQ}\right) \subsetneq \Gamma_{\mathcal{M}}\left(C_3^{\cCQQ}\right)$. 
The approximation is given by the three variable Shannon inequalities and the following additional inequalities:
\begin{align} 
&\!-H(X)\!-\!H(Y)\!-\!H(Z)\!+\!H(XY)\!+\!H(XZ)\!\geq\! 0,\nonumber \\
&-3 H(X)-3 H(Y)- 3 H(Z)+ 2H(XY)\nonumber\\
&+ 3H(XZ)+ 2H(YZ)- H(XYZ)\geq 0,\label{eq:ccq}
\end{align}
up to permutations of $X$, $Y$ and $Z$ in the first inequality and of
$Y$ and $Z$ in the second. Note that these five inequalities are a
subset of the seven inequalities \eqref{eq:margindep} delimiting
$\Gamma_{\mathcal{M}}\left( C_3^{\cC}\right)$ and that
\begin{equation*}
\Gamma_{\mathcal{M}}\left(C_3^{\cC}\right) \subsetneq \Gamma_{\mathcal{M}}\left(C_3^{\cCCQ}\right).
\end{equation*} 

\begin{widetext}
The inequalities $\Diamond_{XZBY} \geq 0$, $\Diamond_{YZAX} \geq 0$ and $\Diamond_{YZXA}\geq 0$
lead to the additional inequalities
\begin{equation}\label{eq:ccqadditional}
\begin{split}
-2H(X)-2H(Y)-2H(Z)+3H(XY)+3H(XZ)+3H(YZ)-4H(XYZ) &\geq 0, \\
-6H(X)-6H(Y)-6H(Z)+5H(XY)+5H(XZ)+5H(YZ)-3H(XYZ) &\geq 0, \\
-4H(X)-4H(Y)-4H(Z)+3H(XY)+4H(XZ)+3H(YZ)-2H(XYZ) &\geq 0,
\end{split}
\end{equation}
\end{widetext}
(including permutations of $X$ and $Y$ in the last inequality).  They
render the second inequality (and its permutations) in \eqref{eq:ccq}
redundant, while the first remains (for all of its permutations). Note
that the first inequality of \eqref{eq:ccqadditional} is also present
in $\Gamma_{\mathcal{M}} \left(C_3^{\cCQQ}\right)$.
As in the previous example, further constraints could likely be derived by considering additional non-Shannon inequalities.

\section{Conclusions} \label{sec:conclusion} 

We have shown that
non-Shannon inequalities tighten the entropic approximations of the
classical entropy cones in many causal structures including the
triangle scenario and the causal structure relevant for information
causality. 
Our newly derived inequalities improve on the entropic distinction of
these from other (classical) causal structures, which is of interest
for inferring (classical) causal relations. They also constitute a set
of restrictions on the classical entropy cones that we cannot derive
in the quantum case, which may point towards differences between the
sets of achievable entropy vectors in classical and quantum case. 

Since it is known from the Bell scenario that quantum correlations can
be detected by considering the entropies of the variables in a post-selected causal structure~\cite{Braunstein1988}, our analysis of the information causality scenario is the one that is most likely to be useful for this purpose.
In this context, non-Shannon inequalities may also be important with regard to the discussion of whether entropic techniques may even be sufficient for certifying classical reproducibility in certain scenarios, a question that has previously been explored for the CHSH scenario in Ref.~\cite{Chaves2013}.

While the entropy vector approach is known to be a useful
means for distinguishing different classical causal structures, its
ability to differentiate between classical and quantum versions of the
same causal structure is known to be limited~\cite{pnpaper}. The
present work has unveiled further limitations of the approach:   
for all causal structures classified
in~\cite{Henson2014} we found either that the sets of achievable
entropy vectors in classical and quantum case coincide (for the causal structures of Figure~\ref{fig:HensonStructures}), or that
non-Shannon inequalities play a role in their characterisation leaving
us unable to make such a statement. 

One of the reasons why it is difficult to make such a statement when
non-Shannon inequalities play a role is our relatively poor
understanding of the structure of entropy space. Even in the absence
of a causal structure we lack a tight characterisation of the set of
allowed entropy vectors for four random variables. In the quantum
case, it is an important open problem whether any further general
constraints on the von Neumann entropy exist. This partly explains our inability to show whether there is some causal structure in which
the described entropy vector approach can be useful for distinguishing
classical and quantum.

Behind all this is the question, of whether there is a novel technique
that allows for an efficient and accurate way to distinguish classical
and quantum versions of the same causal structure. Such a technique
needs to simplify the description of the set of allowed distributions
but remain complex enough to retain the distinctive features of
classical, quantum and post-quantum probability
distributions. Identifying such a quantity would provide further
insight into the meaning of cause in quantum mechanics.

\vspace{-0.4cm}

\acknowledgments We thank Benjamin Chapman and Alex May for their
assistance with initial investigations, and we thank Luke Elliott for
carrying out random searches in Fortran. We are grateful to Matthew
Pusey for alerting us to an error in an earlier draft.  This work was
supported by the Engineering and Physical Sciences Research Council
through a First Grant (no.\ EP/P016588/1) and the Quantum
Communications Hub (grant no.\ EP/M013472/1).

\bibliographystyle{plainnat}

\appendix

\section{Discussion of discrepancy with \cite{Chaves2014,Chaves2015}}\label{sec:prevwork}
Our new approximations to
$\overline{\Gamma^*_\cM}\left(C_3^{\cC}\right)$ presented in
Section~\ref{sec:triangle} contradict the claim
in~\cite{Chaves2014,Chaves2015} that
$\overline{\Gamma^*_\cM}\left(C_3^{\cC}\right)=\Gamma_\cM(C_3^{\cC})$. This
appendix reviews these results and explains the discrepancy.

In~\cite{Chaves2014,Chaves2015}, the inequalities defining the set
$\Gamma_{\mathcal{M}}\left(C_3^{\cC}\right)$ as well as its vertex
description were calculated. Furthermore, probability distributions
$P \in \mathcal{P}_{\mathcal{M}}$ that achieve the rays of
$\Gamma_{\mathcal{M}}\left(C_3^{\cC}\right)$ were presented in the
Supplementary Information of~\cite{Chaves2015}. However, it was not
shown there, that the corresponding distributions, $P$, lie in
$\mathcal{P}_{\mathcal{M}}\left(C_3^{\cC}\right)$, and hence that the
corresponding entropy vectors are achievable in $C_3^{\cC}$.

Our results imply that
$\Gamma_{\mathcal{M}}\left(C_3^{\cC}\right) \neq
\overline{\Gamma^*_{\mathcal{M}}}\left(C_3^{\cC}\right)$
and that three of the extremal rays of
$\Gamma_{\mathcal{M}}\left(C_3^{\cC}\right)$ cannot lie within
$\overline{\Gamma^*_{\mathcal{M}}}\left(C_3^{\cC}\right)$,
specifically the ray containing the vector
$v=\left( 2,3,3,4,4,5,6 \right)$ and its permutations. In the
Supplementary Information of~\cite{Chaves2015}, $v$ is shown to be
achieved with the probability distribution
\begin{equation}\label{eq:as}
\!\!\!P_{\mathrm{XYZ}}(x,y,z)=\begin{cases}
\frac{1}{64}&x\!\oplus\! y\!+\! x\!\oplus\! z\!+\! y\!\oplus\! z\!=\!0\\
0 &\textrm{otherwise},
\end{cases} 
\end{equation}
where $x\in\left\{1,\ldots,4\right\}$,
$y,z\in\left\{1,\ldots,8\right\}$ and $\oplus$ denotes addition modulo
$2$. This means that $P_{\mathrm{XYZ}}(x,y,z)=\frac{1}{64}$ if and
only if either $x$, $y$ and $z$ are all odd or they are all even. This
distribution can be mapped to the perfect correlations
of~\eqref{eq:perfcor} by locally mapping all odd outcomes to $1$ and
all even outcomes to $0$ at $X$, $Y$ and $Z$. Since perfect
correlations are known not to be achievable in
$\mathcal{P}_{\mathcal{M}}\left(C_3^{\cC}\right)$, the
distribution~\eqref{eq:as} is not compatible with the triangle causal
structure.

This resolves the apparent contradiction of our results with those from~\cite{Chaves2014,Chaves2015}.
What was shown there is that
\begin{equation*}
\Gamma_{\mathcal{M}}\left(C_3^{\cC}\right) \subseteq \overline{\Gamma^*_3},
\label{eq:chaves}
\end{equation*}
i.e., that all vectors $v \in \Gamma_{\mathcal{M}}\left(C_3^{\cC}\right)$
can be written as the entropy of a valid probability distribution, or arbitrarily well approximated by such (but
not necessarily one that is achievable in the triangle causal
structure).

\section{Infinite families of inequalities} \label{sec:families}
Infinite families of inequalities may be derived to tighten the
entropic approximation to
$\overline{\Gamma_{\cM}^{*}}\left(C^{\cC}_3\right)$. Here we give the
proof for the three examples provided in
Proposition~\ref{lemma:matus1}.  However, there are numerous other
examples that can be derived in a similar way.

The families~\eqref{eq:marginalmatus1} and~\eqref{eq:mat1fam2} are
derived from~\eqref{eq:matus1} by combining the inequalities for one
$s$-value at a time with Shannon and conditional independence
constraints. These are the only families derived from
\eqref{eq:matus1} in this way for which none of the resulting
inequalities are rendered redundant by those found in the calculations
for Case~2 in Section~\ref{sec:triangle}.

\begin{proof}[Proof of Proposition~\ref{lemma:matus1}]
We tackle the three inequalities separately.

\noindent\eqref{eq:marginalmatus1}: The instance of inequality \eqref{eq:matus1} with $(X_1,~X_2,~X_3,~X_4)=(X,~Y,~Z,~C)$ can be rewritten as 
\begin{align*}
&\left(-\frac{1}{2} s^2 - \frac{3}{2} s \right) H(X)
- \left(s+1\right) H(Y) \\
- &\left(\frac{1}{2} s^2 + \frac{1}{2} s \right) H(Z)
+ s H(CX) \\
+ &s H(CY)
- s H(CZ)
+\left( \frac{1}{2} s^{2} + \frac{3}{2} s + 1 \right) H(XY) \\
+ &\left( s^{2} + 2s \right) H(XZ)
+ \left( \frac{1}{2} s^{2} + \frac{3}{2} s + 1 \right) H(YZ) \\
- &s H(CXY)
- \left( s^2+2s+1 \right) H(XYZ) \geq 0.
\end{align*}
Applying $I(X\!   : \!  Y|C)=0$ and $I(Z\!   : \!  C)=0$, all terms containing the
variable $C$ cancel and we recover~\eqref{eq:marginalmatus1}.\bigskip

\noindent\eqref{eq:mat1fam2}: Inequality~\eqref{eq:matus1}, with variable choices
$(X_1,~X_2,~X_3,~X_4)=(Y,~X,~C,~Z)$ and using the independences $I(X\!   : \!  Y|C)=0$ and $I(Z\!   : \!  C)=0$, can be rewritten as 
\begin{align}
&\left(  \frac{1}{2} s^2 +  \frac{1}{2} s +1  \right) H(C)
- \left(s+1\right) H(X) \nonumber \\
-& \left(\frac{1}{2} s^2 + \frac{3}{2} s \right) H(Y)
-\left( \frac{1}{2} s^{2} + \frac{1}{2} s \right) H(CX) \nonumber \\
- &H(CY) 
+\left(  \frac{1}{2} s^{2} + \frac{3}{2} s + 1  \right) H(XY)-s H(Z)\nonumber \\
+&s H(XZ) 
+ s H(YZ)
- s H(XYZ)\geq 0.  \label{eq:m0}
\end{align}
We also marginalise $\Gamma{\left(C_3^{\cC}\right)}$ to obtain constraints on the vectors 
\begin{multline*}\left( H(C), H(X), H(Y), H(Z), H(CX), H(CY), \right. \\ \left. H(XY), H(XZ), H(YZ), H(XYZ) \right) 
\end{multline*}
that arise from Shannon and independence inequalities. Two of the constraints this elimination yields are the following inequalities,
\begin{align}
-&2H(C)-2H(X)-2H(Y)-3H(Z)+H(CX) \nonumber \\
+&H(CY)+H(XY)+2H(XZ)+2H(YZ)\nonumber \\
-&H(XYZ) \geq 0  , \label{eq:m1} \\\nonumber\\
-&H(C)-H(X)-H(Z)+H(CX)+H(XZ) \! \geq \! 0 . \label{eq:m2} 
\end{align}
We now use~\eqref{eq:m1} to remove $H(CY)$ from \eqref{eq:m0}, which yields
\begin{align}
&\left(\frac{1}{2} s^2 + \frac{1}{2} s - 1 \right) H(C)
 - \left(\frac{1}{2} s^2 + \frac{3}{2} s + 2 \right) H(Y)
 \nonumber\\
 -& \left(s+3 \right) H(X)- \left(  s+3 \right) H(Z) + \left(s +2 \right) H(XZ) 
\nonumber\\
+&\left( \! -\frac{1}{2} s^{2}- \frac{1}{2} s +1 \! \right) H(CX)+ \left(s +2 \right) H(YZ) 
 \nonumber\\
+& \left( \! \frac{1}{2} s^{2} + \frac{3}{2} s + 2 \! \right) \! H(XY)
+ \left(- s +1 \right) \! H(XYZ) \! \geq \! 0  .  \label{eq:m3}
\end{align}
With \eqref{eq:m2}, $H(CX)$ and $H(C)$ are eliminated from
\eqref{eq:m3}, which concludes the proof for this case.\bigskip

\noindent\eqref{eq:mat2}: In a similar manner as for the family~\eqref{eq:mat1fam2}, we consider
inequality~\ref{eq:matus2} with variable choices
$(X_1,~X_2,~X_3,~X_4)=(X,~Y,~C,~Z)$ and the independences $I(X\!   : \!  Y|C)=0$
and $I(Z\!   : \!  C)=0$ to obtain
\begin{align}
&\left( \! s +1 \! \right) \! H(C)
- \! \left( \! \frac{1}{2} s^2 + \frac{3}{2} s \! \right) \! H(X) 
- \! \left( \! s +1 \! \right) \! H(Y) \nonumber\\
-&\left(\frac{1}{2} s^2 + \frac{1}{2} s \right) H(Z)
- H(CX) 
-s H(CY) \nonumber\\
+&\left( \frac{1}{2} s^{2} + \frac{3}{2} s + 1 \right) H(XY) 
+ s^2 H(XZ) \nonumber\\
+&\left( \frac{1}{2} s^{2} + \frac{1}{2} s  \right) H(YZ)
- s^2 H(XYZ)\geq 0 \ .  \label{eq:n0}
\end{align}
We also consider two inequalities that are obtained from marginalising
$\Gamma{\left(C_3^{\cC}\right)}$ to vectors
\begin{multline*}\left( H(C), H(X), H(Y), H(Z), H(CX), H(CY),
  \right. \\ \left. H(XY), H(XZ), H(YZ), H(XYZ) \right)  ,
\end{multline*} namely \eqref{eq:m1} as well as
\begin{equation}
-H(C)-H(Y)-H(Z)+H(CY)+H(YZ) \geq 0  . \label{eq:n1}
\end{equation}
Inequality \eqref{eq:m1} allows us to eliminate $H(CX)$ from \eqref{eq:n0} and \eqref{eq:n1} allows us to eliminate $H(C)$ and $H(CY)$, which concludes the proof.
\end{proof}

\section{Entropy cones for causal structures with up to five variables} \label{sec:uptofive} 
For most causal structures with up to five nodes, the sets of
compatible distributions generated with classical and quantum
resources are identical, and hence, so are their entropic
cones~\cite{Henson2014}. Ref.~\cite{Henson2014} reports one causal
structure with four nodes (Figure~\ref{fig:HensonStructures}(a)) and
$96$ causal structures with five, for which this equivalence does not
hold. These were reduced to the three causal structures shown in
Figure~\ref{fig:HensonStructures}, using reduction criteria. In the
following, we show that, for the three causal structures in question,
the classical and quantum entropy cones coincide. Note that this does
not imply that the same holds true for the remaining $94$ causal
structures. An example where we have not been able to establish this
is the causal structure $\hat{\rm IC}$.

\begin{proof}[Proof of Proposition~\ref{prop:hen}]
  We begin by showing that for the causal structures shown in
  Figure~\ref{fig:HensonStructures}, the classical entropy cones
  coincide with the corresponding Shannon approximations. For the
  instrumental scenario of Figure~\ref{fig:HensonStructures}(a) this
  is shown in Example~\ref{example:innerI}, for the Bell scenario of
  Figure~\ref{fig:HensonStructures}(b) this was previously shown
  in~\cite{pnpaper}. In the following, we hence consider
  Figure~\ref{fig:HensonStructures}(c). 

  The Shannon inequalities and independence constraints lead to an
  outer approximation that is the conic hull of the following vectors,
  denoted here as lists of their components, ordered as
  {\small
  $$( \!  H(X), \! H(Y), \!  H(Z),\!  H(XY),\!  H(XZ),\!  H(YZ),\!  H(XYZ) \! ).$$
  }
\begin{alignat*}{3}
&(1) \quad & &1 1 2  2 2 2  2 \\
&(2) \quad & &1 2 1  2 2 2  2 \\
&(3) \quad & &1 1 1  2 2 2  2 \\
&(4) \quad & &1 2 2  2 2 2  2 \\
&(5) \quad & &1 1 1  2 2 1  2 \\
&(6) \quad & &1 1 1  1 1 1  1 \\
&(7) \quad & &1 0 1  1 1 1  1 \\
&(8) \quad & &1 1 0  1 1 1  1 \\
&(9) \quad & &1 0 0  1 1 0  1 \\
&(10) \quad & &0 1 0  1 0 1  1 \\
&(11) \quad & &0 0 1  0 1 1  1 \\
\end{alignat*}
The following strategies confirm that all of the extremal rays are achievable within the causal structure and, hence, that we have found the associated entropy cone. Note that $\oplus$ denotes addition modulo $2$.
\begin{itemize}
\item The entropy vectors (1) and (2) are recovered by choosing $A$ and $B$ to be uniform bits and $X=A \oplus B$, $Y=B$, $Z=(A, X)$, or $X=A \oplus B$, $Y=(B, X)$, $Z=A$ respectively.
\item (3) is recovered by letting $A$ and $B$ be uniform bits and $X=A \oplus B$, $Y=B$, $Z=A$.
\item The entropy vector (4) is recovered by letting $A$ and $B$ be uniform bits and $X=A \oplus B$, $Y=(B,X)$, $Z=(A,X)$.
\item Let $A$ and $B$ be uniform bits and let $X=A \oplus B$, $Y=B$, $Z=A \oplus X$ to recover (5).
\item $X$ is a uniform bit and $Y=X=Z$ to recover (6).
\item To recover vectors (7) and (8), $A$ or $B$ are taken to be a uniform bit, and $X=A=Z$ or $X=B=Y$ respectively. The remaining variable is deterministic. 
\item Entropy vectors (9)-(11) are obtained by choosing either $X$, $Y$, or $Z$ respectively to be uniform bits and the other two variables to take a value deterministically.
\end{itemize}

We next show that in all three examples the Shannon outer approximation also
coincides with the set of compatible entropy vectors in the quantum
case $\overline{\Gamma^{*}_{\cM}}\left({C^{\qQ}}\right)$.  For this,
we rely on the facet description of the respective cones and show that
each of the inequalities also holds in the quantum case.

1. For the instrumental scenario of Figure~\ref{fig:HensonStructures}(a) the only inequality in addition to the Shannon inequalities for three observed variables is ${I(X : ZY )\leq H(Z)}$~\cite{Henson2014}. This holds in the quantum case because 
\begin{align*}
I(X\!   : \!  ZY) \! &\leq \! I(X\!   : \!  ZA_Y) \\
\! &\leq \!  H(X)+ H(Z) + H(A_Y) - H(XZA_Y) \\
\! &\leq \! H(Z) + H(X A_Y) - H(XZA_Y) \\
\! &\leq \! H(Z), \,
\end{align*}
where the first inequality is a DPI, then we use submodularity, the independence of $X$ and $A_Y$ and monotonicity for the cq-state $\rho_{XZA_Y}$.

2. For the Bell scenario of Figure~\ref{fig:HensonStructures}(b) the only constraints (in addition to the four variable Shannon inequalities) are the independencies $I(W\!   : \!  YZ)=0$ and $I(Z\!   : \!  WX)=0$, which hold in the quantum case.

3. For the causal structure of Figure~\ref{fig:HensonStructures}(c)
the only additional inequality is
$I(Y\!   : \!  Z|X) \leq H(X)$~\cite{Henson2014}. This holds in the quantum case
because
\begin{align*}
I&(Y\!   : \!  Z|X) \\
\! &\leq \! I(Y\!   : \!  B_Z|X) \\
\! &\leq \! I(A_Y\!   : \!  B_Z|X) \\
\! &\leq \! H(A_Y X) \! + \! H(B_Z X) \! - H(X) \! - \! H(A_Y B_Z) \\
\! &= \! H(A_Y X) \! + \! H(B_Z X) \! - \! H(X) \! - \! H(A_Y) \! - \! H(B_Z) \\
\! &\leq \! H(X), \,
\end{align*}
where the first two inequalities are DPIs and the third holds by
monotonicity. The equality holds because $A_Y$ and $B_Z$ are
independent and the last inequality follows from two submodularity
constraints.
\end{proof}

\section{Inner approximation to $\overline{\Gamma^{*}_{\cM}}\left(\hat{\rm IC}^{\cC}\right)$}\label{sec:example}
Computing an outer approximation in terms of Shannon and independence (in)equalities as well as including all permutations of the Ingleton inequality for four of the five variables, yields a cone with $46$ extremal rays. In the following, we list an entropy vector on each such extremal ray, with components 
\begin{equation*}
\begin{split}
(H(X_0), H(X_1), H(Z), H(Y),H(X_0X_1),H(X_0Z),\\H(X_0Y),H(X_1Z),H(X_1Y),H(ZY),H(X_0X_1Z),\\ H(X_0X_1Y),H(X_0ZY),H(X_1ZY),H(X_0X_1ZY) ),
\end{split}
\end{equation*}
where rays that are obtained from others by permuting $X_0$ and $X_1$ are omitted.
\begin{alignat*}{3}
&(1) \quad & &1 1 1 1   1 1 1 1 1 1   1 1 1 1   1 \\
&(2) \quad & &1 1 1 0   1 1 1 1 1 1   1 1 1 1   1 \\
&(3) \quad & &1 0 1 1   1 1 1 1 1 1   1 1 1 1   1 \\
&(4) \quad & &1 1 0 0   1 1 1 1 1 0   1 1 1 1   1 \\ 
&(5) \quad & &1 0 1 0   1 1 1 1 0 1   1 1 1 1   1 \\ 
&(6) \quad & &0 0 1 1   0 1 1 1 1 1   1 1 1 1   1 \\ 
&(7) \quad & &1 0 0 0   1 1 1 0 0 0   1 1 1 0   1 \\  
&(8)\quad & &0 0 1 0   0 1 0 1 0 1   1 0 1 1   1 \\ 
&(9)\quad & &0 0 0 1   0 0 1 0 1 1   0 1 1 1   1 \\ 
&(10)\quad & &1 1 1 1   2 2 2 2 2 1   2 2 2 2   2 \\ 
&(11)\quad & &1 1 1 1   1 2 2 2 2 2   2 2 2 2   2 \\ 
&(12)\quad & &1 1 1 1   1 2 1 2 1 2   2 1 2 2   2 \\ 
&(13)\quad & &1 0 1 1   1 2 1 1 1 2   2 1 2 2   2 \\ 
\end{alignat*}
\begin{alignat*}{3}
&(14)\quad & &1 0 1 1   1 2 2 1 1 2   2 2 2 2   2 \\ 
&(15)\quad & &1 1 1 0   2 2 1 2 1 1   2 2 2 2   2 \\ 
&(16)\quad & &0 1 1 2   1 1 2 2 2 2   2 2 2 2   2 \\ 
&(17)\quad & &1 1 1 2   1 2 2 2 2 2   2 2 2 2   2 \\ 
&(18)\quad & &1 1 2 1   2 2 2 2 2 2   2 2 2 2   2 \\ 
&(19)\quad & &1 1 2 1   2 3 2 3 2 3   3 3 3 3   3 \\ 
&(20)\quad & &1 1 1 2   2 2 3 2 3 2   3 3 3 3   3 \\ 
&(21)\quad & &2 1 2 1   2 3 2 3 2 3   3 2 3 3   3 \\ 
&(22)\quad & &1 1 2 1   2 3 2 2 2 3   3 2 3 3   3 \\ 
&(23)\quad & &1 1 2 1   2 3 2 3 2 3   3 2 3 3   3 \\ 
&(24)\quad & &1 1 1 1   2 2 2 2 2 2   3 3 3 3   3 \\ 
&(25)\quad & &1 1 2 2   2 3 3 3 2 3   3 3 3 3   3 \\ 
&(26)\quad & &1 1 2 2   2 3 3 2 3 3   3 3 3 3   3 \\ 
&(27)\quad & &1 1 2 1   2 3 2 3 1 3   3 2 3 3   3 \\ 
&(28)\quad & &1 1 1 1   2 2 2 2 2 2   3 2 3 3   3 \\ 
&(29)\quad & &1 1 2 2   2 3 3 3 3 3   3 3 3 3   3 \\ 
&(30)\quad & &1 1 2 2   2 3 2 3 2 3   3 2 3 3   3 \\ 
&(31)\quad & &1 1 2 3   2 3 3 3 3 3   3 3 3 3   3 \\ 
&(32)\quad & &1 1 2 2   2 3 3 3 3 4   4 3 4 4   4 \\ 
&(33)\quad & &2 1 2 1   3 4 3 3 2 3   4 3 4 4   4 \\ 
&(34)\quad & &1 1 2 3   2 3 4 3 4 4   4 4 4 4   4.\\ 
\end{alignat*}
We have identified probability distributions compatible with $\hat{\rm IC}^{\cC}$ that reproduce vertices on each of the rays. Hence the convex hull of these rays is an inner approximation to $\overline{\Gamma^{*}_{\cM}}\left(\hat{\rm IC}^{\cC}\right)$. It is characterised by $23$ classes of inequalities, giving a total number of $35$ inequalities when including permutations. In the following we list distributions recovering one vector on each extremal ray (again not listing strategies for the rays that are obtained from others by permuting $X_0$ and $X_1$). For this purpose, let $C_1$, $C_2$, $C_3$, $C_4$, $C_5$ and $C_6$ be random bits and let $\oplus$ denote addition mod~$2$.
\begin{itemize}
\item (1) Let $Y=Z=X_1=X_0=C_1$.
\item (2) Let $Z=X_1=X_0=C_1$ and $Y=1$; 
(3) let $Y=Z=X_0=C_1$ and $X_1=1$. 
\item (4) Let $X_1=X_0=C_1$ and $Y=Z=1$; (5) let $Z=X_0=C_1$, $Y=1$ and $X_1=1$; 
(6) let $Y=Z=C_1$ and $X_1=X_0=1$.
\item (7) Let $X_0=C_1$ and $Y=Z=X_1=1$. (8)-(9) are permutations of this.
\item (10) Let $X_0=C_1$, $X_1=C_2$ and $Y=Z=X_0\oplus X_1$; (11) Let $X_1=X_0=C_1$, let $Y=A=C_2$ and $Z=X_1 \oplus A$.
\item (12) Let $X_1=X_0=C_1$, $A=C_2$, $Z=X_1 \oplus A$ and $Y=Z \oplus A$.
\item (13) Let $X_0=C_1$ and $X_1=1$, let $A=C_2$, $Z=X_0 \oplus A$ and $Y=Z \oplus A$. 
\item (14) Let $X_0=C_1$ and $X_1=1$, let $Y=A=C_2$ and $Z=X_0 \oplus A$; 
(15) let $X_0=C_1$ and $X_1=C_2$, let $Z=X_0 \oplus X_1$ and $Y=1$.
\item (16) Let $X_0=1$ and $X_1=C_1$, let $A=C_2$, $Z=X_1 \oplus A$ and $Y=(A, Z \oplus A)$. 
\item (17) Let $X_0=X_1=C_1$, let $A=C_2$, $Z=X_1 \oplus A$ and $Y=(A, Z \oplus A)$.
\item (18) Let  $X_0=C_1$ and $X_1=C_2$, let $Z=(Z_1,Z_2)=(X_0,X_1)$ and $Y=Z_1 \oplus Z_2$.
\item (19) Let $X_0=C_1$, $X_1=C_2$ and $A=C_3$, let $Z=(X_0\oplus A,X_1 \oplus A)$ and $Y=A$; (20) let $X_0=C_1$, $X_1=C_2$ and $A=C_3$, let $Z=X_0 \oplus X_1 \oplus A$ and $Y=(A, Z \oplus A)$.
\item (21) Let $X_0=(C_1,C_2)$, $X_1=C_1$ and $A=C_3$, let $Z=(Z_1,Z_2)=(C_2\oplus A,X_1 \oplus A)$ and $Y=Z_1 \oplus A$. 
\item (22) Let $X_0=C_1$, $X_1=C_2$ and $A=C_3$, let $Z=(Z_1,Z_2)=(X_0 \oplus A, X_1)$ and $Y=Z_1 \oplus Z_2 \oplus A$. 
\item (23) Let $X_0=(C_1,C_2)$, $X_1=(C_3,C_4)$ and $A=(C_5,C_6)$, let $Z=(Z_1,Z_2,Z_3,Z_4)=(C_1 \oplus C_5, C_2 \oplus C_6, C_3 \oplus C_6, C_4 \oplus C_5 \oplus C_6)$ and $Y=(Z_1 \oplus Z_3 \oplus C_5 \oplus C_6,Z_2 \oplus Z_4 \oplus C_5)$.\footnote{Note that the entropy values here are double the ones given in the description of the extremal ray.}
\item (24) Let $X_0=C_1$, $X_1=C_2$ and $A=C_3$, let $Z=X_0 \oplus  X_1 \oplus A$ and $Y= A$.
\item (25) Let $X_0=C_1$, $X_1=C_2$ and $A=C_3$, let $Z=(Z_1,Z_2)=(X_0 \oplus A, X_1 \oplus A)$ and $Y=(A, Z_2 \oplus A)$; (26) let $X_0=C_1$, $X_1=C_2$ and $A=C_3$, let $Z=(Z_1,Z_2)=(X_0 \oplus A, X_1)$ and $Y=(Z_1, Z_2 \oplus A)$. 
\item (27) Let $X_0=C_1$, $X_1=C_2$ and $A=C_3$, let $Z=(Z_1,Z_2)=(X_0 \oplus A , X_1 \oplus A)$ and $Y= Z_2 \oplus A$. 
\item (28) Let $X_0=C_1$, $X_1=C_2$ and $A=C_3$, let $Z=X_0 \oplus  X_1 \oplus A$ and $Y=Z \oplus A$.
\item (29) Let $X_0=C_1$, $X_1=C_2$ and $A=C_3$, let $Z=(Z_1,Z_2)=(X_0 \oplus A, X_1 \oplus A)$ and $Y=(Z_1 \oplus Z_2 \oplus A,  A)$; (30) Let $X_0=C_1$, $X_1=C_2$ and $A=C_3$, let $Z=(Z_1,Z_2)=(X_0 \oplus A, X_1 \oplus A)$ and $Y=(Z_1 \oplus  A, Z_2 \oplus A)$; (31) let $X_0=C_1$, $X_1=C_2$ and $A=C_3$, let $Z=(Z_1,Z_2)=(X_0 \oplus A, X_1 \oplus A)$ and $Y=(Z_1 \oplus  A, Z_2 \oplus A, A)$.
\item (32) Let $X_0=C_1$, $X_1=C_2$ and $A=(C_3,C_4)$, let $Z=(Z_1,Z_2)=(X_0 \oplus C_3, X_1 \oplus C_4)$ and $Y=(C_3, Z_1 \oplus Z_2 \oplus C_3 \oplus C_4)$; (33) let $X_0=(C_1,C_2)$, $X_1=C_3$ and $A=C_4$, let $Z=(Z_1,Z_2)=(C_1 \oplus X_1, C_2 \oplus A)$ and $Y=Z_1 \oplus Z_2 \oplus A$. 
\item (34) Let $X_0=C_1$, $X_1=C_2$ and $A=(C_3,C_4)$, let $Z=(Z_1,Z_2)=(X_0 \oplus C_3, X_1 \oplus C_4)$ and $Y=(C_3,C_4, Z_1 \oplus Z_2 \oplus C_3 \oplus C_4)$.
\end{itemize}
The searches for these distributions were performed by hand; they could, however, also be straightforwardly automated.
 
The Shannon outer approximation to $\overline{\Gamma^{*}_{\cM}}\left(\hat{\rm IC}^{\cC}\right)$ shares the $46$ extremal rays of the inner approximation (given above) but has six additional ones, where in the following we list one vector on each ray, omitting rays obtained through permutation of $X_0$ and $X_1$ as above,
\begin{alignat*}{3}
&(35)\quad & &2 2 2 2   3 3 3 3 4 3   4 4 4 4   4 \\
&(36)\quad & &2 2 2 2   3 3 3 4 3 3   4 4 4 4   4 \\
&(37)\quad & &2 2 3 2   3 4 3 4 3 5   5 4 5 5   5 \\
&(38)\quad & &2 2 3 2   4 4 3 4 3 4   5 4 5 5   5 .
\end{alignat*}
We can show that these vectors are all outside $\overline{\Gamma^{*}_{\cM}}\left(\hat{\rm IC}^{\cC}\right)$, by resorting to non-Shannon inequalities. The Shannon outer approximation is characterised by $19$ classes of inequalities, or a total of $29$ inequalities including permutations.

\section{Proof of Proposition~\ref{prop:triangleingleton}} \label{sec:inner_appendix}
In the following we prove Proposition~\ref{prop:triangleingleton}.
First, we have computed the $7$ extremal rays of
  $\Gamma_\mathcal{M}^\mathrm{I}\left(C_3^{\cC}\right)$. We list one vector on each such extremal
  ray in the following\footnote{As usual, we order the components as $(H(X),H(Y),H(Z),H(XY),H(XZ),H(YZ),H(XYZ))$.}:
\begin{alignat*}{3}
&(1) \quad & &1 1 1 2 2 2 2 \\
&(2) \quad & &0 0 1 0 1 1 1 \\
&(3) \quad & &0 1 0 1 0 1 1 \\
&(4) \quad & &1 0 0 1 1 0 1 \\
&(5) \quad & &0 1 1 1 1 1 1 \\
&(6) \quad & &1 0 1 1 1 1 1 \\
&(7) \quad & &1 1 0 1 1 1 1 .
\end{alignat*}
These rays can also be analytically shown to be the extremal rays of
$\Gamma_\mathcal{M}^\mathrm{I}\left(C_3^{\cC}\right)$.\footnote{To do
  so, note that in seven dimensions seven inequalities can lead to at
  most seven extremal rays (choosing six of the seven to be
  saturated).  One can then check that each of the claimed rays
  saturates six of the seven inequalities constraining
  $\Gamma_\mathcal{M}^\mathrm{I}\left(C_3^{\cC}\right)$.}  For each
extremal ray we show how to generate a probability distribution
$P\in\mathcal{P}_\mathcal{M}\left(C_3^{\cC}\right)$ whose entropy
vector $v\in\overline{\Gamma^*_\mathcal{M}}\left(C_3^{\cC}\right)$
lies on the ray. To do so, let $A$, $B$ and $C$ be uniform random
bits.
\begin{itemize}
\item (1): Take $X=C$, $Z=A$ and $Y=A \oplus C$ where $ \oplus$ stands for addition modulus $2$.
\item (2): Take $Z=A$ and let $X=1$ and $ Y=1$ deterministic. (3) and
  (4) are permutations of this.
\item (5): Choose $Y=A=Z$ and let $X=1$ deterministic. (6) and (7) are permutations of this.
\end{itemize}
In this way, all extremal rays of
$\Gamma_\mathcal{M}^\mathrm{I}\left(C_3^{\cC}\right)$ are achieved
by vectors in $\overline{\Gamma^*_\mathcal{M}}\left(C_3^{\cC}\right)$ and, by
convexity of $\overline{\Gamma^*_\mathcal{M}}\left(C_3^{\cC}\right)$,
we have $\Gamma_\mathcal{M}^{I}\left(C_3^{\cC}\right)\subseteq
\overline{\Gamma^*_\mathcal{M}}\left(C_3^{\cC}\right)$.

To show that the inclusion is strict, let $A$, $B$ and $C$ be uniform
random bits. Let $X=\operatorname{AND}(B,C)$,
$Y=\operatorname{AND}(A,C)$ and $Z=\operatorname{OR}(A,B)$. The
marginal distribution $P_\mathrm{XYZ} \in
\mathcal{P}_\mathcal{M}\left(C_3^{\cC}\right)$ leads to an entropy
vector $v_1=(0.81,~0.81,~0.81,~1.55,~1.5,~1.5,~2.16)$ and an
interaction information of $I(X\!   : \!  Y\!   : \!  Z) \approx 0.04 > 0$ (where all
numeric values are rounded to two decimal places). Hence, $v_1 \in
\overline{\Gamma^*_\mathcal{M}}\left(C_3^{\cC}\right)$ but $v_1
\notin \Gamma_\mathcal{M}^{I}\left(C_3^{\cC}\right)$ and therefore
$\Gamma_\mathcal{M}^{I}\left(C_3^{\cC}\right)\subsetneq
\overline{\Gamma^*_\mathcal{M}}\left(C_3^{\cC}\right)$. \qed

\section{Proof of Proposition~\ref{prop:andextension}} \label{sec:fritzproof} In this
section, we prove Proposition~\ref{prop:andextension} based on  reasoning from~\cite{Fritz2012}. There, the idea is that in $C_3^{\qQ}$ one can take $X$ and $Y$ to correspond to two bits, which we call
$(\tilde{X},\tilde{B})$ and $(\tilde{Y},\tilde{A})$ respectively.
The quantum state corresponding to node
$C$ is a maximally entangled state $\Psi_\mathrm{C}=
\frac{1}{\sqrt{2}} ( \left| 01 \right\rangle - \left| 10 \right\rangle
)$, the first half of which is the subsystem to $C_X$ and the second
half is $C_Y$. $A$ and $B$ can be taken to be uniform classical
bits.  
We introduce $\Pi_\theta=\ketbra{\theta}{\theta}$, where
$\ket{\theta}=\cos(\frac{\theta}{2})\ket{0}+\sin(\frac{\theta}{2})\ket{1}$,
and the four POVMs
\begin{align*}
&E_0=\left\{\Pi_0,\Pi_{\pi}\right\}&&E_1=\left\{\Pi_{\pi/2},\Pi_{3\pi/2}\right\}\\
&F_0=\left\{\Pi_{\pi/4},\Pi_{5\pi/4}\right\}&&F_1=\left\{\Pi_{3\pi/4},\Pi_{7\pi/4}\right\}.
\end{align*}
Consider a measurement on the $C_X$ subsystem with POVM $E_B$ (i.e.,
if $B=0$ then $E_0$ is measured and otherwise $E_1$), and likewise a
measurement on $C_Y$ with POVM $F_A$.  Let us denote the corresponding
outcomes $\tilde{X}$ and $\tilde{Y}$.  With this choice
$P_{\mathrm{\tilde{X}\tilde{Y}|AB}}$ violates the CHSH
inequality~\cite{Clauser1969}.  The observed variables are then
$X=(\tilde{X},\tilde{B})$, $Y=(\tilde{Y},\tilde{A})$ and $Z=(A',B')$,
with the correlations set up such that $B'=\tilde{B}=B$ and
$A'=\tilde{A}=A$.
In essence the reason that this cannot be realised in the causal
structure $C_3^{\cC}$ is the CHSH violation. Note though that it is
also important that information about $A$ is present in both $Y$ and
$Z$ (and analogously for $B$).  If for example, we consider the same
scenario but with $Y=\tilde{Y}$ then we could mock-up the correlations
classically.  This can be done by removing $A$, replacing $B$ with
$(B_1,B_2)$ and taking $B_1,B_2$ and $C$ to each be a uniform random
bit.  We can then take $Y=C$, $Z=(B_1,B_2)$ and
$X=(f(C,B_1,B_2),B_1)$, where $f$ is chosen appropriately. Since $f$
can depend on all of the other observed variables it can generate any
correlations between them\footnote{This is like playing the CHSH game
  but where Alice knows Bob's input and output in addition to her own
  input.}. In the causal structure $C_3$, taking $\tilde{A}=A'$
ensures that these are shared through $A$ and hence information about
them cannot be used to generate $X$.

Our Proposition~\ref{prop:andextension} requires a restriction of $Z$ to one bit of information, which we prove to be possible in the following.

\begin{proof}[Proof of Proposition~\ref{prop:andextension}]
First, since all classical distributions can be realised using quantum
systems, $\mathcal{P}_\mathcal{M}\left(C_3^{\cC}\right) \subseteq
\mathcal{P}_\mathcal{M}\left(C_3^{\qQ}\right)$.  We now show that this
inclusion is strict in the case where two nodes output two bits and
one node outputs only one. To do so, we consider the setup in
Figure~\ref{fig:fritzproof} of Section~\ref{sec:quantcorr}, taking
$Z=\operatorname{AND}(A,B)$, $\tilde{A}=A$, $\tilde{B}=B$ and
$P_\mathrm{\tilde{X}\tilde{Y}|\tilde{A}\tilde{B}}$ to violate the CHSH
inequality. This yields an observed distribution of the form
\begin{equation} \label{eq:violatingdistr}
P_\mathrm{\tilde{A}\tilde{B}\tilde{X}\tilde{Y}Z}
=\frac{1}{4}P_\mathrm{\tilde{X}\tilde{Y}|\tilde{A}
  \tilde{B}}\delta_{Z,\operatorname{AND}(\tilde{A},\tilde{B})}\, .
\end{equation}
We start with the following lemma which identifies the causes of
$\tilde{A}$ and $\tilde{B}$ for any distribution of the
form~\eqref{eq:violatingdistr} in the triangle causal structure.

\begin{lemma}\label{lemma:andlemma}
If
$P_\mathrm{(\tilde{X}\tilde{A})(\tilde{Y}\tilde{B})Z}\in\mathcal{P}_\mathcal{M}\left(C_3^{\cC}\right)$
obeys~\eqref{eq:violatingdistr}, then
$P_\mathrm{\tilde{A}|A C}=P_\mathrm{\tilde{A}|A}$ and $P_\mathrm{\tilde{B}|B C}=P_\mathrm{\tilde{B}|B}$.
\end{lemma}

\begin{proof}
Due to the causal constraints we can write 
\begin{equation*}
P_\mathrm{Z \tilde{A} \tilde{B}}=\sum_{abc} P_\mathrm{Z|ab}P_\mathrm{\tilde{A}|ac} P_\mathrm{\tilde{B}|bc} P_\mathrm{A}(a) P_\mathrm{B}(b) P_\mathrm{C}(c).
\end{equation*}
Because $Z=\operatorname{AND}(\tilde{A}, \tilde{B})$, we can derive the following two conditions:
\begin{enumerate}
\item Using $P_\mathrm{Z \tilde{A} \tilde{B}}(0,1,1)=0$, it follows
  that for each triple $(a,b,c)$ either $P_\mathrm{Z|ab}(0)=0$ or
  $P_\mathrm{\tilde{A}|ac}(1)=0$ or $P_\mathrm{\tilde{B}|bc}(1)=0$.
\item Using $P_\mathrm{Z\tilde{A}\tilde{B}}(1,0,0)=P_\mathrm{Z
    \tilde{A} \tilde{B}}(1,0,1)=P_\mathrm{Z \tilde{A}
    \tilde{B}}(1,1,0)=0$, it follows that for each triple
  $(a,b,c)$ either $P_\mathrm{Z|ab}(1)=0$ or
  $P_\mathrm{\tilde{A}|ac}(1)=P_\mathrm{\tilde{B}|bc}(1)=1$.
\end{enumerate}
We first argue that $Z$ is a deterministic function of $A$ and $B$,
i.e., $P_{Z|ab}\in\{0,1\}$ for all pairs $(a,b)$. From condition 2 we
know that either $P_\mathrm{Z|ab}(1)=0$ (and thus
$P_\mathrm{Z|ab}(0)=1$) deterministically, or that
$P_\mathrm{\tilde{A}|ac}(1)=P_\mathrm{\tilde{B}|bc}(1)=1$. But in the
latter case condition 1 implies that $P_\mathrm{Z|ab}(0)=0$ (and
thus $P_\mathrm{Z|ab}(1)=1$).

Now let us consider the two cases (a)~$P_\mathrm{Z|ab}(1)=1$; and
(b)~$P_\mathrm{Z|ab}(1)=0$ separately.\smallskip

\noindent (a) Let $(a,b)$ be such that
$P_\mathrm{Z|ab}(1)=1$. According to condition 2,
$P_\mathrm{\tilde{A}|ac}(1)=P_\mathrm{\tilde{B}|bc}(1)=1$ for all $c$,
and thus we have $P_\mathrm{\tilde{A}|ac}=P_\mathrm{\tilde{A}|a}$ as
well as $P_\mathrm{\tilde{B}|bc}=P_\mathrm{\tilde{B}|b}$.\smallskip

\noindent (b) Let $(a,b)$ be such that $P_\mathrm{Z|ab}(1)=0$. Then $P_\mathrm{Z|ab}(0)=1$ and thus by condition 1 for
all $c$ either $P_\mathrm{\tilde{A}|ac}(1)=0$ or
$P_\mathrm{\tilde{B}|bc}(1)=0$.  We further divide into two cases:
either (i) $(a,b)$ are such that $P_\mathrm{Z|ab'}(1)=0$ for all
$b'$ and $P_\mathrm{Z|a'b}(1)=0$ for all $a'$; or (ii) they are
not. \smallskip

\noindent (ii) Suppose $\exists b'$ such that $P_\mathrm{Z|ab'}(1)=1$.
In this case $P_\mathrm{\tilde{A}|ac}(1)=1$ for all $c$ due to
condition 2 and thus from condition 1 we have $P_\mathrm{\tilde{B}|bc}(1)=0$. Thus for such pairs $(a,b)$, the relations
$P_\mathrm{\tilde{A}|ac}=P_\mathrm{\tilde{A}|a}$ as well as
$P_\mathrm{\tilde{B}|bc}=P_\mathrm{\tilde{B}|b}$ hold. Symmetric
considerations can be made in the case where $\exists a'$ such that
$P_\mathrm{Z|a'b}(1)=1$ instead.\smallskip

\noindent (i) It cannot be the case that all pairs $(a,b)$
have $P_\mathrm{Z|ab'}(1)=0$ for all $b'$ and $P_\mathrm{Z|a'b}(1)=0$ for all $a'$ (otherwise $P_Z(1)=0$). Hence there exists
$(a'',b'')$ for which $P_\mathrm{Z|a''b''}(1)=1$. By
condition 2, this implies that $P_\mathrm{\tilde{A}|a''c}(1)=P_\mathrm{\tilde{B}|b''c}(1)=1$ for all $c$. Thus, as $P_\mathrm{Z|ab''}(1)=0$ and $P_\mathrm{\tilde{B}|b''c}(1)=1$, it follows from
condition 1 that $P_\mathrm{\tilde{A}|ac}(1)= 0$ for any $c$;
$P_\mathrm{\tilde{B}|bc}(1)=0$ follows analogously, which concludes
the proof.
\end{proof}

To prove the proposition, we will suppose that $P_\mathrm{(\tilde{X}\tilde{A})(\tilde{Y}\tilde{B})Z}\in\mathcal{P}_\mathcal{M}\left(C_3^{\cC}\right)$ and derive a contradiction.

First note that the previous lemma together with the form of
$C_3^{\cC}$ implies
\begin{equation}\label{eq:anoc}
P_\mathrm{\tilde{A}|C}=P_{\mathrm{\tilde{A}}},\text{ and }
P_\mathrm{\tilde{B}|C}=P_{\mathrm{ \tilde{B}}}\, .
\end{equation}

Furthermore, from
$P_\mathrm{(\tilde{X}\tilde{A})(\tilde{Y}\tilde{B})Z}\in\mathcal{P}_\mathcal{M}\left(C_3^{\cC}\right)$
we have
\begin{eqnarray*}
P_\mathrm{\tilde{A}\tilde{B}\tilde{X}\tilde{Y}}=\sum_cP_\mathrm{C}(c)P_\mathrm{\tilde{A}|c}P_\mathrm{\tilde{B}|c}P_\mathrm{\tilde{X}\tilde{Y}|\tilde{A}\tilde{B}c}
\end{eqnarray*}
which, using~\eqref{eq:anoc}, and the form of $C_3^{\cC}$ can be
rewritten
$$P_\mathrm{\tilde{A}\tilde{B}\tilde{X}\tilde{Y}}=\sum_cP_\mathrm{C}(c)P_\mathrm{\tilde{A}}P_\mathrm{\tilde{B}}P_\mathrm{\tilde{X}|\tilde{B}c}P_\mathrm{\tilde{Y}|\tilde{A}c}
\, .$$
However, that $P_\mathrm{\tilde{X}\tilde{Y}|\tilde{A}\tilde{B}}$ violates a
Bell inequality means that this equation cannot hold, establishing a
contradiction.
\end{proof}

\end{document}